\documentclass{amsproc}
\usepackage[margin=2cm]{geometry}
\newtheorem{theorem}{Theorem}[section]

\newtheorem{lemma}[theorem]{Lemma}

\theoremstyle{definition}

\theoremstyle{remark}

\numberwithin{equation}{section}
\newcommand{\commentout}[1]{}

\numberwithin{equation}{section}

\newcommand{\underhat}[1]{\underset{\widehat{}}{#1}}

\begin{document}

\title{Connection preserving deformations and $q$-semi-classical orthogonal polynomials}
\author{Christopher M. Ormerod, N. S. Witte and Peter J. Forrester}
\address{Department of Mathematics and Statistics. The University of Melbourne Parkville VIC 3010 Australia}%


\begin{abstract}
We present a framework for the study of $q$-difference equations satisfied by $q$-semi-classical orthogonal systems. 
As an example, we identify the $q$-difference equation satisfied by a deformed version of the little $q$-Jacobi polynomials as a gauge transformation of a special case of the associated linear problem for $q$-$\mathrm{P}_{\mathrm{VI}}$. We obtain a parameterization of the associated linear problem in terms of orthogonal polynomial variables and find the relation between this parameterization and that of Jimbo and Sakai.
\end{abstract}

\maketitle

\section{Introduction}

Monodromy representations have been a central element in the study of integrable systems \cite{ClassicalIntegrability}. 
A pioneering step was the parameterization of the condition that a linear second order differential equation with four regular singularities $\{0,t,1,\infty\}$ has monodromy independent of $t$, in terms of the sixth Painlev\'e equation
$\mathrm{P}_{\mathrm{VI}}$ \cite{Fuchs1907, Fuchs1914}. This theory was elaborated upon by Garnier \cite{Ga_1917}
and Schlesinger \cite{Schlesinger}, and culminated in the 1980's with the studies of the Kyoto School \cite{JimboMiwa1, JimboMiwa2, JimboMiwa3}. A contemporary perspective of the theory can be found in the monographs \cite{IKSY_1991} and \cite{FIKN_2006}. 
For a matrix linear differential equation of the form
\begin{equation}\label{eq0:linear}
\frac{\mathrm{d}}{\mathrm{d}x}Y(x) = A(x) Y(x),
\end{equation}
where
\[
A(x) = \sum_i \frac{A_i}{x-\alpha_i},
\]
one expects the general solution to be multivalued with branch points located at $\alpha = \{\alpha_i\}$. By evaluating a solution on any element of the homotopy classes of closed loops, $[\gamma]$, in some manner around a selection of the poles, one obtains the equation 
\[
Y(\gamma(1)) = Y(\gamma(0)) M_{[\gamma]}. 
\]
This relates solutions on different sheets of a Riemann surface. The set $\{ M_{[\gamma]} : \gamma:[0,1]\to \mathbb{C}\}$ is a representation of the fundamental group of the compliment of the poles, $\Gamma = \pi_1\left( \mathbb{CP}_1 \setminus \{ \alpha \} \right)$. The aim is to deform the linear system with respect to a chosen deformation parameter, $t$, so that the representation of  $\Gamma$ does not depend on $t$. In the theory of monodromy preserving deformations, a natural choice of parameters are the poles of $A$. This leads to the classical Schlesinger's equations \cite{Schlesinger}
\begin{align*}
\frac{\partial A_i}{\partial \alpha_j} =& \frac{[A_i, A_j]}{\alpha_i - \alpha_j},\qquad i \neq j,\\
\frac{\partial A_i}{\partial \alpha_i} =& -\sum_{j \neq i} \frac{[A_i, A_j]}{\alpha_i - \alpha_j}.
\end{align*}
Given a $2\times 2$ linear system with four poles, $\{0,t,1,\infty\}$, and zero of $A_{12}(x)$ at $y$, imposing the isomonodromic property in the variable $t$ requires that $y$ satisfies $\mathrm{P}_{\mathrm{VI}}$ \cite{Fuchs1907}.

The relationship of the theory of monodromy preserving deformations to orthogonal polynomials arises because under certain
conditions on the orthogonality measure the polynomials and their associated functions form an isomonodromic system,
albeit one with a particular restriction. 
Under fairly general conditions the derivative of each polynomial in the system is expressible in terms of linear combinations of other members of the orthogonal polynomial system, an observation first made by Laguerre \cite{Laguerre}. 
The conditions for when this is the case have been given in some generality by Bonan and Clark\cite{BC_1990}, and by
Bauldry \cite{Ba_1990}. 
In particular, a semi-classical weight will satisfy these conditions and the three term recurrence tells us that we may express the derivative of the polynomial in the system of orthogonal polynomials as a rational linear combination of the polynomial itself and the previous polynomial in the system. The rationality of this linear problem means that such orthogonal polynomial systems satisfy a linear problem of the form \eqref{eq0:linear}. 

The notion of a semi-classical weight or linear functional was introduced by Maroni \cite{Ma_1985} as an attempt
to characterize the classical orthogonal polynomials in a coherent framework and guide the quest of looking for
systems beyond this class.
By appropriately extending the work of Laguerre, Magnus \cite{Magnus1995} was able to show that a semi-classical, deformed orthogonal polynomial system \eqref{eq0:linear} parameterized a special case of the monodromy preserving deformation considered by Fuchs \cite{Fuchs1907}. This allows one to express special solutions of $\mathrm{P}_{\mathrm{VI}}$ in terms of coefficients of orthogonal polynomial systems. Conversely, this also allows key quantities relating to
 orthogonal polynomials to be expressed in terms of solutions of $\mathrm{P}_{\mathrm{VI}}$. In addition to the various determinantal solutions of integrable systems provided by the theory of orthogonal polynomials, the application of integrable systems to orthogonal polynomials have resulted in advances in the calculation of various statistics of interest in random matrix theory (see e.g.~\cite{Forrester2007}).
  
For $q$-difference equations, an analogue of the theory of monodromy preserving deformations is the theory of connection preserving
deformations \cite{Sakai:qP6}. The linear problem of interest is given by the $n\times n$ matrix equation 
\begin{equation}\label{eq0:qlinear}
Y(qx) = A(x) Y(x),
\end{equation}
where
\[
A(x) = A_0 + A_1x + \ldots + A_mx^m.
\]
Instead of the basic information being contained in the relation between the value of solutions on different sheets of the Riemann surface, the variables of interest are associated with the relation between two fundamental solutions, $Y_0$ and $Y_\infty$, which are holomorphic functions at $0$ and $\infty$ respectively. Much of the theory concerning the existence of these solutions has remained relatively unchanged since the pioneering days of Birkhoff and his followers \cite{Adams:General, Birkhoff:General, Carmichael:General}. If these solutions exist, then one may meromorphically continue these solutions on $\mathbb{C}$, and furthermore form the connection matrix, $P(x)$, specified by
\[
Y_0(x) = Y_\infty(x) P(x),
\]
which is quasi-periodic in $x$ \cite{Birkhoff:General}. From the Galois theory of $q$-difference equations, which primarily considers a classification of problems of the form \eqref{eq0:qlinear}, we know that the entries of $P(x)$ are expressible in terms of elliptic theta functions \cite{Sauloy:Galois}. 

In the same manner as monodromy preserving deformations, one may consider a deformation of \eqref{eq0:qlinear} that preserves the connection matrix. 
An appropriate choice of deformation parameter turns out to be the roots of the determinant of $A$ and the eigenvalues of $A_0$ and $A_m$. By considering a $2\times 2$ linear system with $m = 2$ and choosing the deformation parameter to be proportional to two of the roots of the determinant and the two eigenvalues of $A_0$, Jimbo and Sakai \cite{Sakai:qP6} showed that the connection preserving deformation was equivalent to a second order $q$-difference equation admitting the sixth Painlev\'e equation as a continuum limit. 

We have remarked that semi-classical orthogonal polynomial systems give rise to monodromy preserving
deformations relating to Painlev\'e equations. A natural problem then is to investigate the relationship between
$q$-semi-classical orthogonal polynomial systems, connection preserving deformations, and the $q$-Painlev\'e equations. 
A number of different approaches to constructing isomonodromic analogues for the difference, $q$-difference and elliptic
equations of the Sakai Scheme \cite{Sakai} have been proposed recently \cite{AB06,AB07,Rains,Yamada}, 
which differ in varying degrees from what we offer here. One other work which is close to the spirit of the present
work is that of Biane \cite{Biane:qP6}. However there is a history of studies into $q$-semi-classical orthogonal polynomial 
systems which was not motivated by the above considerations. Shortly after the introduction of the semi-classical
concepts Magnus extended this to the $q$-difference systems and in fact to the most general type of
divided difference operators on non-uniform lattices in a pioneering study \cite{Magnus88}. 
In addition Maroni and his co-workers have made extensions to difference and $q$-difference systems in a 
series of works \cite{KM_2002,MM_2002,Kh_2003,Me_2009,GK_2009}. These later authors have successfully reproduced 
parts of the classical Askey Tableaux (which was achieved most fully by Magnus at the level of the Askey-Wilson
polynomials) however the application of their theoretical tools beyond the classical cases have invariably been 
made to specialised or degenerate cases and failed to make contact with the discrete and $q$-Painlev\'e
equations. A slightly different methodology has been the approach of Ismail and collaborators
\cite{IW_2001,Is_2003,Chen:qladder,IS_2009}, who have derived difference and $q$-difference equations for orthogonal
polynomials with respect to weights more general than the semi-classical class, much in the spirit of the
Bonan and Clark and Bauldry studies, and so the matrix $ A(x) $ is no longer rational. This approach has not
been applied to systems beyond the classical Askey Tableaux, and consequently not made contact with the 
discrete Painlev\'e systems. The most recent work of Biane \cite{Biane:qP6}, and of Van Assche and co-workers \cite{VanAssche,BSvA_2008}
has addressed some of the shortcomings discussed above, however while these authors have uncovered the 
spectral structures of the theory they have yet to elucidate the deformation structures required. It is
our intention to complete this task by laying out the deformation structures.   

Our contributions in this paper are to first formulate an extension of the classical work of
Laguerre for finding differential equations satisfied by orthogonal polynomials,
when the differential operator is the $q$-difference
\begin{equation}\label{eq0:Dqdef}
D_{q,x_i} f(x_1, \ldots, x_n) = \frac{ f(x_1,\ldots, x_n) - f(x_1, \ldots, qx_i, \ldots x_n)}{x_i(1-q)}.
\end{equation}
This is done in Section 4, after first having introduced preliminary material from orthogonal polynomial
theory, and connection preserving deformations in Sections 2 and 3 respectively.
In Section 5 we apply this extension to the specific case of a deformation of the
little $q$-Jacobi polynomials \cite{Andrews:orthogonal}. We give a parameterization of the associated 
linear problem in terms of variables relating to the orthogonal polynomial system. However, this contains 
redundant variables, and in fact a set of three natural coordinates can be identified which suffice to parameterize the linear problem.
When written in terms of the natural coordinates, the linear problem implies the $q$-$\mathrm{P}_{\mathrm{VI}}$
equations
\begin{subequations}\label{eq0:qP6}
\begin{gather}
y_1\hat{y}_1 = \frac{a_7a_8 (y_2 - a_1t)(y_2 - a_2t)}{(y_2 - a_3)(y_2-a_4)},\\
y_2\hat{y}_2 = \frac{a_3a_4 (\hat{y}_1 - a_5t)(\hat{y}_1 - a_6t)}{(\hat{y}_1 - a_7)(\hat{y}_1-a_8)},
\end{gather}
\end{subequations}
where $y_i = y_i(t)$ and $\hat{y}_i = y_i(qt)$.
We show that this has the consequence of implying the $\tau$--functions have determinantal solutions in
terms of Hankel determinants of the moments of the little $q$-Jacobi weight. Also we show that the three
term recurrence, written in terms of the natural coordinates, manifests itself as a
B\"acklund transformation which relate to a translational component of the extended affine Weyl group of type $D_5^{(1)}$. 
Throughout we shall assume that $q$ is a fixed complex number such that $0 < |q| < 1$. 

\section{Orthogonal polynomials}

Our starting point is a sequence of moments, $\{\mu_k \}_{k=0}^{\infty}$. From this we define a linear functional, $L$, on the space of polynomials, where $L(x^k) = \mu_k$. An orthogonal polynomial system is a sequence of polynomials, $\{ p_n \}_{n= 0}^{\infty}$, such that $p_m$ is a polynomial of exact degree $m$ and 
\begin{align}\label{eq1:LinearForm}
L(p_ip_j) = \delta_{ij}.
\end{align}
In other words, these polynomials are orthonormal with respect to the given linear functional. This condition defines the coefficients of $p_n$ for all $n$ so long as the Hankel determinants consisting of the moments $\mu_0, \ldots, \mu_{2n}$, given in \eqref{eq2:Delta},
do not vanish \cite{Chihara, Szego}. 

In the case of the classical continuous orthogonal polynomials, this linear functional, $L$, is typically some integral of the multiplication of the argument with some weight function over some support. Linear functionals associated with discrete orthogonal polynomials are specified by a weighted sum, such as Jackson's $q$-integral \cite{Hahn, Jackson:integral, Thomae1, Thomae2}. Any orthogonal polynomial system where \eqref{eq1:LinearForm} holds satisfies the classical three term recurrence relation, given by
\begin{equation}\label{eq1:3termrecurrence}
a_{n+1}p_{n+1} = (x-b_n)p_n - a_n p_{n-1}.
\end{equation}
We parameterize the coefficients of these polynomials by 
\[
p_n(x) = \gamma_n x^n + \gamma_{n,1} x^{n-1} + \gamma_{n,2} x^{n-2} + \ldots + \gamma_{n,n}.
\]
It is possible to determine all the coefficients, and hence the $a_n$ and $b_n$, in terms of the $\mu_k$'s
 \cite{Chihara, Szego}. We set 
\begin{equation}\label{eq2:Delta}
\Delta_n = \det 
\begin{pmatrix} 
\mu_0 & \mu_1 & \ldots & \mu_{n-1} \\
\mu_1 & \mu_2 & \ldots & \mu_n \\
\vdots & & \ddots & \vdots \\
\mu_{n-2} & \mu_{n-1} & & \vdots \\ 
\mu_{n-1} & \mu_n & \ldots & \mu_{2n-2}
\end{pmatrix},
\end{equation}
for $n \geq 1$, with $\Delta_0 = 1$, and
\begin{equation}\label{eq2:Sigma}
\Sigma_n = \det\begin{pmatrix} \mu_0 & \mu_1 & \ldots & \mu_{n-2} & \mu_{n} \\
\mu_1 & \mu_2  & \ldots & \mu_{n-1} & \mu_{n+1} \\
\vdots & \ddots & \ddots & \vdots & \vdots \\
\mu_{n-1} & \mu_n & \ldots & \mu_{2n-3} & \mu_{2n-1}
\end{pmatrix} ,
\end{equation}
for $n\geq 1$, where $\Sigma_0 = 0$. Then we have
\begin{subequations}\label{eq1:anandbndets}
\begin{align}
\label{eq2:detgamma}\gamma_n^2 =& \frac{\Delta_n}{\Delta_{n+1}},\\
\label{eq2:deta}a_n^2 =& \frac{\Delta_{n-1}\Delta_{n+1}}{\Delta_n^2},\\
\label{eq2:detb}b_n =& \frac{\Sigma_{n+1}}{\Delta_{n+1}} - \frac{\Sigma_n}{\Delta_n},
\end{align}
\end{subequations}
as given in \cite{Chihara, Szego}.

Given a sequence of valid moments, one may define the Stieltjes function
\[
f = \sum_{n = 0}^{\infty} \mu_n x^{-n-1}.
\]
We define the associated polynomials and associated functions by the formula
\[
f p_n = \phi_{n-1} + \epsilon_n,
\]
where $\phi_{n-1}$ is a polynomial and $\epsilon_n$ is the remainder. The orthogonality condition implies that
$\epsilon_n \sim \gamma_n^{-1} x^{-n-1} + O(x^{-n-2})$. In fact, by using \eqref{eq1:3termrecurrence}, it is possible to find the
large $x$ expansions for these polynomials in terms of the $a_n$ and $b_n$, giving  
\begin{subequations}\label{eq1:largexexpansion}
\begin{align}
p_n =& \gamma_n\left(x^n - x^{n-1}\sum_{i =0}^{n-1} b_i + x^{n-2}\left(\sum_{i =0}^{n-2} \sum_{j=i+1}^{n-1} b_i b_j - \sum_{i=1}^{n-1}a_i^2\right) + O(x^{n-3}) \right),\\  
\epsilon_n =& \gamma_n^{-1}\left(x^{-n-1} + x^{-n-2}\sum_{i =0}^n b_i + x^{-n-3}\left(\sum_{i =0}^{n} \sum_{j=0}^{i} b_i
b_j + \sum_{i=1}^{n+1}a_i^2\right) +  O(x^{-n-4}) \right),
\end{align}
\end{subequations}
where in the second equation use has been made of the large $x$ expansion of 
$\gamma_n\epsilon_n = (x-b_n)^{-1} \gamma_{n-1} \epsilon_{n-1} + (x-b_n)^{-1} a_{n+1}^2 \gamma_{n+1} \epsilon_{n+1}$ \cite{Ismail:Book, Chihara, Szego, Magnus1995}. Utilising this we can write some explicit relations between the coefficients, $\gamma_n$ and $\gamma_{n,k}$'s, and the $a_n$ and $b_n$,
\begin{subequations}
\begin{align}
a_n =& \frac{\gamma_{n-1}}{\gamma_n},\\
b_{n-1} =&  \frac{\gamma_{n-1,1}}{\gamma_{n-1}}- \frac{\gamma_{n,1}}{\gamma_n},
\end{align}
\end{subequations}
which hold for $n \geq 1$. By equating $(fp_n)p_{n-1}$ with $(fp_{n-1})p_n$ we have the relation
\begin{equation}\label{eq3:anpnen}
\phi_{n-1}p_{n-1} - \phi_{n-2}p_n = \epsilon_{n-1}p_n - \epsilon_np_{n-1} = \frac{1}{a_n},
\end{equation}
which is polynomial by the left hand side and where the final equality follows from \eqref{eq1:largexexpansion}.

It is clear that the sequence of functions, $\{ \epsilon_n \}_{n=0}^{\infty}$, is a solution of \eqref{eq1:3termrecurrence} that is independent of $\{p_n\}$. We find it convenient to introduce the matrix
\begin{equation}\label{eq2:system}
Y_n = \begin{pmatrix}
p_n & \epsilon_n/w \\
p_{n-1} & \epsilon_{n-1}/w
\end{pmatrix}.
\end{equation}
The three term recursion relation is equivalent to the relation
\begin{equation}\label{eq1:linearn}
Y_{n+1} = M_n Y_n,
\end{equation}
where
\[
M_n = \begin{pmatrix}
{\displaystyle \frac{(x-b_n)}{a_{n+1}}} & - {\displaystyle \frac{a_n}{a_{n+1}}} \\
1 & 0
\end{pmatrix}.
\]

\section{Connection preserving deformations}

In this section we revise the established classical theory of systems of linear $q$-difference equations
\cite{Adams:General, Birkhoff:General, Carmichael:General}. The general theory concerns the $m \times m$ matrix system
\begin{equation}\label{eq2:linear}
Y(qx) = A(x)Y(x),
\end{equation}
where $A(x)$ is rational in $x$. We call $A$ the coefficient matrix of the linear $q$-difference equation. One may easily verify that such an equation possesses two symbolic solutions, namely, the infinite products
\begin{subequations}\label{eq4:symbolicsols}
\begin{align}
 &A(x/q) A(x/q^2) A(x/q^3) \ldots, \\
 &A(x)^{-1} A(xq)^{-1} A(xq^2)^{-1} \ldots, 
\end{align}
\end{subequations}
which do not converge in general. We may suitably transform the problem so that $A(x)$ is polynomial, which we parameterize by writing 
\[
A(x) = A_0 + A_1x + \ldots +A_nx^n.
\]
The matrices $A_0$ and $A_n$ are assumed to be semisimple with eigenvalues $\rho_1, \ldots , \rho_m$ and $\lambda_1, \ldots, \lambda_m$ respectively. Regarding the solutions of \eqref{eq2:linear}, we present the following theorem due to Carmichael \cite{Carmichael:General}. 

\begin{theorem}
Suppose the eigenvalues of $A_0$ and $A_n$ satisfy the condition
\[
\frac{\rho_i}{\rho_j} , \frac{\lambda_i}{\lambda_j} \notin \{q , q^2 , q^3, \ldots\},
\]
then there exists two solutions, called the fundamental solutions, of the form 
\begin{align*}
Y_0 =& \widehat{Y}_0 x^{D_0}, \\
Y_\infty =& \widehat{Y}_{\infty} q^{\frac{nu(u-1)}{2}} x^{D_\infty},
\end{align*}
where $\widehat{Y}_0$ and $\widehat{Y}_{\infty}$ are holomorphic functions in neighbourhoods of $x=0$ and $x=\infty$ respectively and $D_0$ and $D_\infty$ are $\mathrm{diag}(\log_q \rho_i)$ and $\mathrm{diag}(\log_q \lambda_i)$ respectively and $u = \log_q x$. 
\end{theorem}

We may use \eqref{eq2:linear} to continue both solutions meromorphically over $\mathbb{C}\setminus 0$. From these solutions, we define the connection matrix to be
\begin{equation}\label{eq2:connectiondef}
P(x) = Y_\infty(x)^{-1} Y_{0}(x).
\end{equation}
The evolution of $ P9x) $ in $x$ is given by
\begin{align}\label{8.1}
P(qx) =& Y_\infty(qx)^{-1} Y_{0}(qx)\nonumber,\\
 =& Y_\infty(x)^{-1}A(x)^{-1}A(x) Y_{0}(x)\nonumber,\\
 =& P(x).
\end{align}
Hence this function is $q$-periodic in $x$. 

These fundamental solutions may be related to \eqref{eq4:symbolicsols} via a conjugation of transformations of \eqref{eq2:linear} such that the solutions given by \eqref{eq4:symbolicsols} converge. Hence \eqref{eq4:symbolicsols} gives us information regarding the roots of the determinant of the connection matrix. If the $z_i$ are the zeros of $\det A(x)$, then $Y_{\infty}^{-1}$ and $Y_0$ are possibly singular at $\{ q^{n+1} z_i : n \in \mathbb{N}\}$ and $\{q^{-n} z_i : n \in \mathbb{N}\}$ respectively. Therefore we expect the poles and zeros of the determinant of $P(x)$ to be $q$-power multiples of the $z_i$. 

In the situation of monodromy preserving deformations we introduce a parameter, $t$, into $A$ and consider what conditions on the evolution of $t$ are required so that the monodromy representation is preserved. It was the innovation of Jimbo and Sakai \cite{Sakai:qP6} to introduce a parameter $t$ in a manner that preserves the connection matrix, $P(x)$, through the evolution $t \to q t$. By the observation noted in the above paragraph regarding the poles and roots of the determinant of $P(x)$, we may infer that the latter are preserved in the shift $z_j \to qz_j$. However, in doing this, we must also consider the eigenvalues of $A_0$ and $A_n$ to be parameters. This suggests that the zeros of the determinant of $A$ and eigenvalues of $A_0$ and $A_n$ are appropriate choices of parameter for the connection matrix preserving deformation. 

We parameterize $\rho_A = \det A$ by letting a subset of the roots be constant in $t$, while the other roots are simply proportional to $t$. By supposing $P(x,t) = P(x,qt)$ we arrive at the implication
\[
P(x,qt) = Y_\infty(x,qt)^{-1} Y_{0}(x,qt) = Y_\infty(x,t)^{-1} Y_{0}(x,t) = P(x,t),
\]
which defines a matrix, $B$, via 
\[
Y_0(x,qt)Y_0(x,t)^{-1} = Y_{\infty}(x,qt) Y_\infty(x,t)^{-1} = B(x,t).
\]
Since both $Y_0$ and $Y_\infty$ are independent solutions this leads to the necessary condition that the evolution of $Y$ must be governed by a second linear equation, 
\begin{equation}\label{eq2:qlineart}
Y(x,qt) = B(x,t)Y(x,t).
\end{equation}
Conversely, it is easy to see that if $Y$ satisfies \eqref{eq2:qlineart}, then the evolution $t\to q t$ defines a connection preserving deformation.

Since $Y$ satisfies an equation in $x$ and an equation in $t$ this imposes the necessary compatibility condition
\begin{equation}\label{eq3:Lax}
A(x,qt)B(x,t) = B(qx,t)A(x,t).
\end{equation}
This compatibility condition implies a $q$-difference equation satisfied by $\rho_B = \det B$,
\[
\rho_{B}(qx,t) = \frac{\rho_{A}(x,qt)}{\rho_{A}(x,t)}\rho_B(x,t),
\]
which may be solved up to a factor of a function of $t$. We shall assume that $\rho_B$ is rational in $x$. Other information regarding asymptotic behavior of $A$ gives us specific information regarding the form of $B$. Furthermore the compatibility condition and the determinantal constraints often results in an overdetermined system allowing us to construct a representation of $B$ in terms of the entries of $A$ and hence, find $q$-difference equations in $t$ for the entries of $A$.

We do not pursue this line explicitly. Rather, in the following sections we will show how the $q$-difference equations satisfied by a particular deformed $q$-semiclassical orthogonal polynomial system leads linear systems satisfying \eqref{eq2:linear} and \eqref{eq2:qlineart}.

\section{$q$-difference equations satisfied by orthogonal polynomials}

\subsection{The $q$-difference calculus and $q$-special functions}

A reference for the $q$-difference calculus and also the $h$-differential calculus is a book by Kac \cite{QuantumCalc}. 
We first recall some of the basic properties in relation to $q$-difference equations. The first property is the $q$-analog of the product and quotient rule, given by
\begin{subequations}\label{eq2:prodquot}
\begin{align}
D_{q,x} (f g) =& \overline{f}D_{q,x}g + g D_{q,x}f, \\
 =& fD_{q,x}g + \overline{g}D_{q,x}f\nonumber, \\
D_{q,x} \left( \frac{f}{g}\right) =& \frac{gD_{q,x}f - fD_{q,x}g}{g\overline{g}},
\end{align}
\end{subequations}
where $\overline{f} = \overline{f(x)} = f(qx)$. Associated with the $q$-difference operator is the antiderivative, known as Jackson's $q$-integral, as defined by Thomae and Jackson \cite{Jackson:integral, Thomae1, Thomae2}. We express the definite integral of Thomae \cite{Thomae1,Thomae2} by
\begin{equation}\label{eq2:Tomae}
\int_0^1 f(x)\mathrm{d}_q x = (1-q)\sum_{k=0}^{\infty} q^k f(q^k).
\end{equation}
This was subsequently generalized by Jackson \cite{Jackson:integral} to 
\begin{equation}\label{eq2:Jacksons}
\int_a^b f(x)\mathrm{d}_q x = \int_0^b f(x)\mathrm{d}_q x - \int_0^a f(x)\mathrm{d}_q x,
\end{equation}
where
\[
\int_0^a f(x)\mathrm{d}_q x = a(1-q)\sum_{k=0}^{\infty} q^k f(aq^k).
\]
If $f$ is continuous, then 
\[
\lim_{q\to 1} \int_0^a f(t)\mathrm{d}_q t = \int_0^a f(t)\mathrm{d}t.
\]
We introduce the $q$-Pochhammer symbol
\[
(a;q)_n = (1-a)(1-aq) \ldots (1-aq^{n-1}),
\]
and
\[
(a;q)_\infty = (1-a)(1-aq)\ldots.
\]
We also adopt the notation
\[
(a_1,a_2, \ldots, a_m;q)_n = (a_1;q)_n (a_2;q)_n \ldots (a_m;q)_n.
\]
The symbol $(a;q)_\infty$ is often called the $q$-exponential as
\[
D_{q,x} \left(ax;q \right)_\infty = \frac{-a(aqx;q)_\infty }{(1-q)}.
\]
We are now able to express Heine's basic hypergeometric function \cite{Heine:Hyper1}, as it was re-written by Thomae \cite{Thomae1, GasperRahman:Hypergeometric},
\begin{equation}\label{eq2:Hypergeometricsum}
{}_2\phi_1 \left(\left.\begin{array}{c} a,b \\ c \end{array}\right|q;t\right) = \sum_{m=0}^{\infty} t^m \frac{(a,b;q)_m}{(c,q;q)_m}.
\end{equation}
Also relevant is the integral formula for the basic hypergeometric function \cite{GasperRahman:Hypergeometric}
\begin{equation}\label{eq2:Hypergeometricintegral}
{}_2\phi_1 \left(\left.\begin{array}{c} a,b \\ c \end{array}\right|q;t\right) = \frac{\left(b,\frac{c}{b};q\right)_\infty}{(1-q)(c,q;q)_\infty}\int_0^1 x^{\frac{\log b}{\log q}-1} \frac{(xta,xq;q)_\infty}{\left(xt,\frac{xc}{b};q\right)_\infty} \mathrm{d}_q x.
\end{equation}

Jacobi's elliptic multiplicative theta function, as defined by Jacobi's triple product formula, may be expressed as
\[
\theta_q(z) = \left( q,-qz,-\frac{1}{z};q\right)_{\infty}.
\]
This satisfies the equation
\[
\theta_q(qz) = qz\theta_q(z). 
\]
Of importance is the $q$-character,
\[
e_{q,c}(x) = \frac{\theta_q (x)\theta_q(1/c)}{\theta_q(x/c)},
\]
satisfying
\begin{equation}\label{eq4:eprops}
e_{q,c}(qx) = c e_{q,c} (x), \qquad e_{q,qc}(x) = xe_{q,c}(x).
\end{equation}

\subsection{$q$-difference equation in $x$}

We shall henceforth assume that the log $q$-derivative of the weight specifying the linear form of the
$q$-orthogonal polynomial system is rational, parameterizing its $q$-derivative via the
equation 
\begin{equation}\label{eq2:qdifofweight}
W(x) D_{q,x} w(x) = 2V(x) w(x),
\end{equation}
where $W(x)$ and $V(x)$ are polynomials in $x$. This is the $q$-analogue of the notion of semi-classical weight functions 
familiar in the theory of classical orthogonal polynomials \cite{Ma_1985} and which was proposed by Magnus\cite{Magnus88}
and subsequently by others \cite{MM_2002,KM_2002}. 
Henceforth we regard systems satisfying these conditions as $q$-semi-classical orthogonal polynomial systems.

\begin{lemma}\label{lem:Dqxf}
Assuming $w$ satisfies \eqref{eq2:qdifofweight}, the Stieltjes function, $f$, satisfies the
 $q$-difference equation
\begin{equation}\label{eq2:StiltjerqDiffx}
W(x)D_{q,x}f(x) = 2V(x)f(x)+ U(x),
\end{equation}
where $U$ is polynomial such that $\deg U < \deg V$.
\end{lemma}

We call the polynomials $W$, $V$ and $U$ the spectral data polynomials. We will use \eqref{eq2:StiltjerqDiffx} as the basis for our derivation of the $x$-evolution for the $q$-semi-classical orthogonal systems. Using different methods, this  
has been derived in \cite{IW_2001,Chen:qladder,Is_2003,IS_2009} for general (i.e. beyond semi-classical) 
$q$-orthogonal polynomials.

\begin{theorem}\cite{Magnus88,Kh_2003}\label{thm:Dqxpn}
The matrix $Y_n$ satisfies
\begin{equation}\label{eq2:linearx}
D_{q,x}Y_n = A_n Y_n,
\end{equation}
where
\begin{align}\label{eq4:Aform}
A_n &= \frac{1}{(W(x)-2x(1-q)V(x))}\begin{pmatrix}
\Omega_n - V & -a_n \Theta_n \\
a_n \Theta_{n-1} & \Omega_{n-1} - V - (x-b_{n-1})\Theta_{n-1}
\end{pmatrix},
\end{align}
and $\Theta_n$ and $\Omega_n$ are polynomials specified by
\begin{subequations}\label{eq2:defThetaOmega}
\begin{align}
\label{eq2:defTheta}\Theta_n =& W (\epsilon_n D_{q,x} p_n - p_n D_{q,x} \epsilon_n) + 2V\epsilon_n\overline{p}_n,\\
\label{eq2:defOmega}\Omega_n =& a_n W (\epsilon_{n-1} D_{q,x} p_n - p_{n-1} D_{q,x} \epsilon_n) \\
  & \hspace{.5cm} + a_nV (p_n\epsilon_{n-1} + p_{n-1}\epsilon_n) - 2V x(1-q) a_n \epsilon_{n-1} D_{q,x} p_n. \nonumber
\end{align}
\end{subequations}
\end{theorem}
\begin{proof}
Using \eqref{eq2:StiltjerqDiffx} and 
\[
f = \frac{\phi_{n-1}}{p_n} + \frac{\epsilon_n}{p_n},
\]
we find
\[
W D_{q,x} \left(\frac{\phi_{n-1}}{p_n}\right) -  \frac{2V\phi_{n-1}}{p_n} - U = -W D_{q,x}\left( \frac{\epsilon_n}{p_n} \right) + \frac{2V\epsilon_n}{p_n}.
\]
This may be expanded to give
\[
\frac{ W(p_n D_{q,x}\phi_{n-1} - \phi_{n-1} D_{q,x} p_n) - 2V\phi_{n-1} \overline{p}_n - Up_n\overline{p}_n}{p_n\overline{p}_n} = \frac{W (\epsilon_n D_{q,x} p_n - p_n D_{q,x} \epsilon_n) + 2V\epsilon_n\overline{p}_n}{p_n\overline{p}_n}.
\]
Setting this equal to $\Theta_n/(p_n \overline{p}_n)$ defines $\Theta_n$ as
\begin{align}
\label{eq4:thetapolyform}\Theta_n =& W(p_n D_{q,x}\phi_{n-1} - \phi_{n-1}D_{q,x} p_n ) - 2V \phi_{n-1}\overline{p}_n - Up_n\overline{p}_n,\\
 =& W (\epsilon_n D_{q,x} p_n - p_n D_{q,x} \epsilon_n) + 2V\epsilon_n\overline{p}_n\nonumber.
\end{align}
The first expression is a linear combination of polynomials which verifies that $\Theta_n$ is a polynomial for all $n$. The second expression is one in which every occurrence of a $p_n$ or its derivative is multiplied by a factor of $\epsilon_n$ or its derivative. Hence, by \eqref{eq1:largexexpansion}, there exists an upper bound for the degree of the polynomial $\Theta_n$ in $x$ that is independent of $n$. 

Using the fact that $\Theta_n$ is polynomial and 
\begin{equation}\label{eq4:anpnrelation}
a_n\phi_{n-1}p_{n-1} - a_n\phi_{n-2}p_n = 1,
\end{equation}
we have
\[
(a_n\phi_{n-1}p_{n-1} - a_n\phi_{n-2}p_n)\Theta_n = W(p_n D_{q,x}\phi_{n-1} - \phi_{n-1}D_{q,x} p_n ) - 2V \phi_{n-1}\overline{p}_n - Up_n\overline{p}_n.
\]
By appropriately equating factors divisible by $\phi_{n-1}$ on one side and factors divisible by $p_n$ on the other side, we let $\Omega_n$ be the common factor by writing 
\begin{align*}
p_n \phi_{n-1} \Omega_n =& \phi_{n-1}\left( a_n\Theta_n p_{n-1} + (W - 2Vx(1-q))D_{q,x} p_n + Vp_n \right),\\
 =& p_n \left( a_n \phi_{n-2} \Theta_n + WD_{q,x} \phi_{n-1} - V \phi_{n-1} - U \overline{p}_n\right).
\end{align*}
The first equality is a rearrangement of the required $q$-difference equation for $p_n$. The second expression for $\Omega_n$ is equivalent to
\begin{equation}\label{eq3:diffphi}
\Omega_n = \frac{a_n\phi_{n-2}\Theta_n}{\phi_{n-1}} + \frac{W D_{q,x} \phi_{n-1}}{\phi_{n-1}} - V - \frac{U\overline{p}_n}{\phi_{n-1}}.
\end{equation}
We use the expression for $\Theta_n$ in terms of $\phi_{n-1}$ and $p_n$ to give
\begin{align*}
\Omega_n =&  \frac{ W a_n\phi_{n-2}p_n D_{q,x} \phi_{n-1} }{\phi_{n-1}} - \frac{ W a_n\phi_{n-2}\phi_{n-1}D_{q,x} p_n }{\phi_{n-1}} - \frac{ 2V a_n\phi_{n-2}  \phi_{n-1}\overline{p}_n }{\phi_{n-1}}\\
 &  - \frac{ U a_n\phi_{n-2}p_n\overline{p}_n }{\phi_{n-1}}  + \frac{W D_{q,x} \phi_{n-1}}{\phi_{n-1}} - V - \frac{U\overline{p}_n}{\phi_{n-1}},
\end{align*}
which, by using the equality $a_n p_n\phi_{n-2} = a_n p_{n-1} \phi_{n-1}-1$ to eliminate occurrences of $\phi_{n-2}$, is equivalent to
\[
\Omega_n = a_n W (p_{n-1} D_{q,x} \phi_{n-1} - \phi_{n-2} D_{q,x} p_n)  - V( 2a_n\phi_{n-2}  \overline{p}_n + 1)  - U a_n p_{n-1}\overline{p}_n.
\]
This expresses $\Omega_n$ as a linear combination of polynomials and hence $\Omega_n$ is a polynomial. 
To obtain the large $x$ expansion, by dividing the first expression for $\Omega_n$ by $\phi_{n-1}p_n$ we have
\begin{equation}\label{eq3:omegaworking}
\Omega_n = \frac{a_np_{n-1}\Theta_n}{p_n}  + \frac{WD_{q,x} p_n}{p_n} + V - \frac{2x(1-q)VD_{q,x}p_n}{p_n}.
\end{equation}
Using \eqref{eq2:defTheta}, we find that $\Omega_n$ may be written as
\begin{equation*}
\Omega_n = \frac{a_np_{n-1}W\epsilon_n D_{q,x} p_n}{p_n}- \frac{a_n
Wp_{n-1}p_n D_{q,x}\epsilon_n}{p_n} +  \frac{2a_np_{n-1}V \epsilon_n  \overline{p}_n}{p_n}+
  \frac{WD_{q,x} p_n}{p_n} + V - \frac{2x(1-q)VD_{q,x}p_n}{p_n}.
\end{equation*}
Using \eqref{eq4:anpnrelation} to cancel the factors of $p_{n-1}$ gives \eqref{eq2:defOmega}.

A rearrangement of \eqref{eq3:omegaworking} is 
\[
(W - 2x(1-q)V)D_{q,x} p_n = (\Omega_{n}-V)p_n - a_n \Theta_n p_{n-1},
\]
where $\Omega_n$ and $\Theta_n$ are given by \eqref{eq2:defTheta} and \eqref{eq2:defOmega}. Mapping $ n \to n-1 $ in this relation and expressing $p_{n-2}$ in terms of $p_n$ and $p_{n-1}$ gives
\[
(W - 2x(1-q)V)D_{q,x} p_{n-1} = a_n\Theta_{n-1}p_n + (\Omega_{n-1} - V - (x-b_{n-1})\Theta_{n-1}) p_{n-1},
\]
showing that the orthogonal polynomials form a column vector solution of \eqref{eq2:linearx}.

To see that $\epsilon_n/w$ satisfies the same $q$-difference equation, we consider $W D_{q,x} \phi_{n-1} + WD_{q,x}\epsilon_n= WD_{q,x}(fp_n)$, which we reformulate as
\begin{align*}
WD_{q,x}\phi_{n-1} + WD_{q,x} \epsilon_n =& W\overline{p}_n D_{q,x} f + W f D_{q,x} p_n, \\ 
=& 2Vf\overline{p}_n + U\overline{p}_n + f\left( (\Omega_n-V)p_n - a_n\Theta_np_{n-1} + 2Vp_n - 2V\overline{p}_n\right), \\
=& \left[ U\overline{p}_n + (\Omega_n + V) \phi_{n-1} - a_n\Theta_n\phi_{n-2}\right] + \left[ (\Omega_n + V) \epsilon_n - a_n\Theta_n \epsilon_{n-1} \right].
\end{align*}
Hence by subtracting $WD_{q,x}\phi_{n-1}$, defined by \eqref{eq3:diffphi}, we find
\[
W D_{q,x} \epsilon_n = (\Omega_n+V)\epsilon_n - a_n\Theta_n \epsilon_{n-1},
\]
and from this
\begin{align*}
D_{q,x}\left( \frac{\epsilon_n}{w}\right) =& \frac{wD_{q,x}\epsilon_n - \epsilon_n D_{q,x} w}{w \overline{w}},\\
 =& \frac{ \frac{\displaystyle (\Omega_n+V)\epsilon_n - a_n\Theta_n \epsilon_{n-1}}{\displaystyle W}
           - \frac{\displaystyle 2V\epsilon_n}{\displaystyle W}}{w\left(1-2x(1-q)\frac{\displaystyle 2V}{\displaystyle W}\right)},\\
 =& \frac{(\Omega_n - V) \frac{\displaystyle \epsilon_n}{\displaystyle w} 
            - a_n\Theta_n \frac{\displaystyle \epsilon_{n-1}}{\displaystyle w}}{W - 2x(1-q)V},
\end{align*}
while the shift $n \to n-1$ and \eqref{eq1:3termrecurrence} gives the compatible evolution equation for $\epsilon_{n-1}/w$.
\end{proof}

Since $p_n$ is of degree $n$ and the leading order term of $\epsilon_n$ about $ x=\infty $ is $x^{-n-1}$, we immediately obtain upper bounds for the degrees of $\Theta_n$ and $\Omega_n$ \cite{Magnus88} 
\begin{subequations}\label{eq3:degthetaomega}
\begin{align}
\label{eq3:degtheta}\deg{\Theta_n} \leq& \max(\deg W-1, \deg V -2 ,0),\\
\label{eq3:degomega}\deg{\Omega_n} \leq& \max(\deg W, \deg V -1 ,0).
\end{align}
\end{subequations}
It follows that knowledge of $W$ and $V$ may be used in conjunction with \eqref{eq1:largexexpansion} to determine $\Omega_n$ and $\Theta_n$ in terms of sums and products of the $a_i$ and $b_i$.

We remark that \eqref{eq2:linearx} can be rewritten in a manner more familiar in the context of connection matrices where $Y_n$
is a solution of the linear $q$-difference equation
\begin{equation}\label{eq4:altdiffx}
Y_n(qx) = (I - x(1-q)A_n)Y_n(x) = \tilde{A}_nY_n(x).
\end{equation}
As revised in Section 3, under appropriate conditions on $\tilde{A}_n$ established by Carmichael \cite{Carmichael:General}, or the more general conditions of Adams \cite{Adams:General}, this equation permits two fundamental solutions: $Y_{\infty,n}$ and $Y_{0,n}$.
We form the connection matrix by using the fundamental solutions of \eqref{eq2:linearx} in \eqref{eq2:connectiondef} 
\begin{equation}\label{eq4:orthconnection}
P_n(x) = (Y_{\infty,n}(x))^{-1} Y_{0,n}(x).
\end{equation}
Solutions of \eqref{eq2:linearx} also satisfy \eqref{eq1:linearn} and so
the transformation $n \to n+1$ is a connection preserving deformation. This also implies that the connection matrix is independent of the consecutive polynomials chosen in the orthogonal system. Therefore the connection matrix is an invariant of the orthogonal polynomial system.

In further specifying the evolution in $x$ for $Y_n$, we have a compatibility relation between the evolution in $x$ and the evolution in $n$. This compatibility relation, defined by equating the two ways of evaluating $D_{q,x}Y_{n+1}$, namely $D_{q,x} M_n Y_n = A_{n+1} M_n Y_n$, results in the condition
\begin{equation}\label{eq3:compxn}
M_n(qx)A_n(x) + D_{q,x}M_n(x) = A_{n+1}(x)M_n(x).
\end{equation}
The entries of the first row are equivalent to the recurrence relations of Magnus \cite{Magnus88,Kh_2003}
\begin{subequations}\label{eq3:Fruedx}
\begin{align}
(x-b_n)(\Omega_{n+1} - \Omega_n) =& W - x(1-q)(\Omega_n+V) - a_n^2 \Theta_{n-1} + a_{n+1}^2 \Theta_{n+1}, \\
\Omega_{n-1} - \Omega_{n+1} =& (x-b_{n-1})\Theta_{n-1} - (qx-b_n)\Theta_n.
\end{align}
\end{subequations}
Using the relation \eqref{eq3:anpnen} we find 
\begin{equation}\label{eq3:determinantofYn}
\det Y_n = \frac{p_n\epsilon_{n-1} - p_{n-1}\epsilon_n}{w}= \frac{1}{a_nw}.
\end{equation}
From this we may deduce
\begin{equation}\label{eq3:detAntilde}
\det(I - x(1-q)A_n) = \frac{\det Y_n(qx)}{\det Y_n(x)} = \frac{w(x)}{w(qx)} = \frac{W(x)}{W(x) - 2x(1-q)V(x)},
\end{equation}
or equivalently
\begin{equation}
x(1-q)\det A_n - \mathrm{Tr} A_n = \frac{2V}{W-2x(1-q)V},
\end{equation}
and these in turn imply the additional recurrence relation 
\begin{equation}\label{eq3:detrelinomega}
x(1-q)\left(\Omega_n^2 - V^2 - a_n^2 \Theta_n \Theta_{n-1}\right) = ((x-b_{n-1})\Theta_{n-1} - \Omega_{n-1} - \Omega_n)(W - x(1-q)(\Omega_n+V)).
\end{equation}

The compatibility between the evolution in $x$ and $n$ for more general orthogonal polynomials has given rise to associated linear problems for discrete Painlev\'e equations \cite{Nijhoff, VanAssche}. Many of these associated linear problems are differential-difference systems \cite{Its}. That is to say that the evolution in $x$ is defined by a differential equation, while the evolution of $n$ is discrete. The first occurrence of a discrete Painlev\'e equation in the literature is thought to have 
been deduced in this manner \cite{Shohat:Orth}.

The $2\times 2$ linear problem derived for orthogonal polynomials is one in which the coefficient matrix, $\tilde{A}_n(x)$, is rational. If we follow the theory of connection matrices, we apply a transformation that relates the linear problem in which $\tilde{A}_n$ is rational to another linear problem in which the coefficient matrix is polynomial. With respect to Birkhoff theory and \eqref{eq2:linear}, the coefficient matrix obeys the  proportionality constraint
\[
\det{A} \propto W(W - 2x(1-q)V).
\]
That is to say that when referring to the connection matrix for orthogonal polynomial systems, we do not distinguish between the roots and the poles of the determinant of the linear problem. 

\subsection{$q$-difference equation in $t$}

We now turn to the new direction that we wish to present. It is at this point we let $w(x) = w(x,t)$, and hence consider polynomials which are both functions of $x$ and $t$. For functions $f(x,t)$ with two independent parameters we will adopt the notation
\begin{gather*}
\overline{f} = \overline{f(x,t)} = f(qx,t),\\
\hat{f} = \widehat{f(x,t)} = f(x,qt).
\end{gather*}
We distinguish a base case in which $\deg \Theta_n = 0$ and $\deg \Omega_n = 1$, corresponding to $\deg V = 1$ and $\deg W = 2$, as being completely solvable and a case in which the connection matrix is known \cite{LeCaine1943}. However by suitably adjoining $q$-exponential factors that depend simply on $t$ to the numerator or denominator of $w(x)$ we introduce roots or poles into $w(qx)/w(x)$. This has the effect of increasing up the degree of $W(x)$ and $V(x)$. Furthermore it imposes a rational character on the logarithmic $q$-derivative of $w$ with respect to $t$ :
\begin{equation}\label{eq3:weightrationaltderive}
R(x,t) D_{q,t} w(x,t) = 2 S(x,t) w(x,t),
\end{equation}
where $R(x,t)$ and $S(x,t)$ are polynomials in $x$. These cannot be arbitrary polynomials in $x$ as there is an implied compatibility condition. This arises because there are two ways of calculating the mixed derivatives of $w$, namely $D_{q,x}D_{q,t}w(x,t)$ and $D_{q,t}D_{q,x} w(x,t)$, the equality of which imposes the constraint
\begin{equation}\label{eq3:compWVRS}
\frac{2\hat{V}}{\hat{W}} \frac{2S}{R} - \frac{2\overline{S}}{\overline{R}}\frac{2V}{W} = D_{q,x} \frac{2S}{R} - D_{q,t} \frac{2V}{W}.
\end{equation}
A consequence of \eqref{eq3:weightrationaltderive} is the following companion result to Lemma \ref{lem:Dqxf}.

\begin{lemma}\label{lem:Dqtf}
The Stieltjes function, $f$, satisfies the $q$-difference equation
\begin{equation}\label{eq2:StiltjerqDifft}
RD_{q,t}f = 2Sf+ T,
\end{equation}
where $T(x,t)$ is polynomial such that $\deg_x T < \deg_x S$.
\end{lemma}

Another compatibility relation is implied by $D_{q,t}D_{q,x} f = D_{q,x}D_{q,t} f$ in conjunction with \eqref{eq2:StiltjerqDiffx} and \eqref{eq2:StiltjerqDifft}. This relation can be stated as
\begin{equation}\label{eq4:Tcomp}
\frac{2\hat{V}}{\hat{W}}\frac{T}{R} - \frac{2\overline{S}}{\overline{R}}\frac{U}{W} = D_{q,x} \frac{T}{R} - D_{q,t} \frac{U}{W}.
\end{equation}
When \eqref{eq3:compWVRS} and \eqref{eq4:Tcomp} are satisfied a companion result to Theorem \ref{thm:Dqxpn} can be stated.

\begin{theorem}\label{thm:Dqtpn}
The matrix $Y_n$ is a solution to
\begin{equation}\label{eq2:lineart}
D_{q,t}Y_n = B_n Y_n,
\end{equation}
where
\begin{align}\label{eq4:Bform}
B_n =\frac{1}{(R - 2t(1-q)S)} \begin{pmatrix}
\Psi_n - S & -a_n \Phi_n \\
a_n \Phi_{n-1} & \Psi_{n-1} - S - (x-b_{n-1})\Phi_{n-1}
\end{pmatrix},
\end{align}
with $\Phi_n$ and $\Psi_n$ polynomials in $x$ specified by
\begin{subequations}\label{eq2:defPhiPsi}
\begin{align}
\label{eq2:defPhi}\Phi_n =&  R\left( \epsilon_n D_{q,t} p_n - p_n D_{q,t} \epsilon_n \right)  + 2S \epsilon_n \hat{p}_n ,\\
\label{eq2:defPsi}\Psi_n =& a_n R_n ( \epsilon_{n-1} D_{q,t} p_n - p_{n-1} D_{q,t} \epsilon_n ) + S ( 2a_n\epsilon_{n-1} \hat{p}_n - 1).
\end{align}
\end{subequations}
\end{theorem}

\begin{proof}
Our strategy is to adapt the proof of Theorem \ref{thm:Dqxpn} to the $D_{q,t}$ operator.
Using $R D_{q,t} f = 2 S f + T$ and $f = \frac{\phi_{n-1}}{p_n} + \frac{\epsilon_n}{p_n}$ we have
\[
R D_{q,t} f = R D_{q,t}\left( \frac{\phi_{n-1}}{p_n} + \frac{\epsilon_n}{p_n}\right) = 2S \frac{\phi_{n-1} + \epsilon_n}{p_n} + T.
\]
This suggests that we define 
\begin{align}
\label{eq4:Phipolyform}\Phi_n =& R\left( p_n D_{q,t} \phi_{n-1} - \phi_{n-1} D_{q,t} p_n \right)  - 2S\phi_{n-1} \hat{p}_n - T p_n \hat{p}_n, \\ 
 =& R\left( \epsilon_n D_{q,t} p_n - p_n D_{q,t} \epsilon_n \right)  + 2S \epsilon_n \hat{p}_n. \nonumber
\end{align}
Expression \eqref{eq4:Phipolyform} tells us $\Phi_n$ is a polynomial in $x$ while \eqref{eq2:defPhi}
implies a bound on the degree.

Using \eqref{eq2:defPhi} and \eqref{eq4:anpnrelation} we arrive at 
\[
a_n(p_{n-1}\phi_{n-1} - p_n \phi_{n-2})\Phi_n = R\left( p_n D_{q,t} \phi_{n-1} - \phi_{n-1} D_{q,t} p_n \right)  - 2S\phi_{n-1} \hat{p}_n - T p_n \hat{p}_n.  
\]
By splitting this expression into terms divisible by $\phi_{n-1}$ and $p_n$, we arrive at an equality that defines $\Psi_n$, given by
\begin{align*}
\phi_{n-1}p_n \Psi_n =& \phi_{n-1}\left( a_n p_{n-1} \Phi_n + (R - 2t(1-q)S)D_{q,t} p_n + S p_n \right), \\
 =& p_n \left( a_n \phi_{n-2} \Phi_n + RD_{q,t} \phi_{n-1} - S \phi_{n-1} - T \hat{p}_n \right).
\end{align*}
The first line is just a rearrangement of the required $q$-difference equation, in $t$, for $p_n$. The second expression is equivalent to
\begin{align*}
\Psi_n =&  \frac{a_n \phi_{n-2} \Phi_n}{\phi_{n-1}} + \frac{RD_{q,t} \phi_{n-1}}{\phi_{n-1}} - S - \frac{T \hat{p}_n}{\phi_{n-1}},\\
 =& \frac{Ra_n \phi_{n-2}p_nD_{q,t} \phi_{n-1}}{\phi_{n-1}} - \frac{Ra_n \phi_{n-2}\phi_{n-1}D_{q,t}p_n}{\phi_{n-1}} - \frac{2Sa_n \phi_{n-2}\phi_{n-1}\hat{p}_n}{\phi_{n-1}}\\ 
  & - \frac{Ta_n \phi_{n-2}p_n \hat{p}_n }{\phi_{n-1}} + \frac{RD_{q,t}\phi_{n-1}}{\phi_{n-1}} - S - \frac{T \hat{p}_n}{\phi_{n-1}},\\
 =& a_nR(p_{n-1} D_{q,t}\phi_{n-1} - \phi_{n-2}D_{q,t}p_n) - S(2a_n\phi_{n-2} \hat{p}_n - 1) - a_nT p_{n-1} \hat{p}_n.
\end{align*}
We remark that this, being a linear combination of polynomials, implies $\Psi_n$ is a polynomial in $x$. 
Using \eqref{eq2:defPhi} in the first expression for $\Psi_n$ allows us to write 
\begin{equation*}
 \Psi_n =  \frac{R a_n p_{n-1}\epsilon_n D_{q,t} p_n}{p_n} - R a_n p_{n-1} D_{q,t} \epsilon_n +
  \frac{2S a_n p_{n-1}\epsilon_n \hat{p}_n}{p_n} + \frac{RD_{q,t}p_n}{p_n} - \frac{2t(1-q)S D_{q,t} p_n}{p_n} + S ,
\end{equation*}
which upon noting $a_np_{n-1}\epsilon_n = a_n p_n \epsilon_{n-1} - 1$ implies \eqref{eq2:defPsi}. 

The working to date shows
\[
(R- 2t(1-q)S)D_{q,t}p_n = (\Psi_n - S)p_n - a_n\Phi_n p_{n-1}.
\]
Replacing $n$ by $n-1$ in this expression, then using
\eqref{eq1:3termrecurrence} to express $p_{n-2}$ in terms of $p_n$ and $p_{n-1}$, establishes 
that $p_n$ and $p_{n-1}$ form a column vector solution of \eqref{eq2:lineart}. 

The derivation of the $q$-difference equation in $t$ for $\epsilon_n / w$ may also be derived in an analogous manner to the proof of Theorem \ref{thm:Dqxpn}, so we refrain from the giving the details.
\end{proof}

We note that $D_{q,t}$ does not necessarily alter the degree in $x$, hence, the upper bounds for the degrees of $\Phi$ and $\Psi$ are given by 
\begin{subequations}\label{eq3:degPhiPsi}
\begin{align}
\deg_x \Phi_n \leq& \max(\deg_x S-1, \deg_x R-1,0 ), \\
\deg_x \Psi_n \leq& \max(\deg_x S, \deg_x R,0 ).
\end{align}
\end{subequations}
Further to this, we may use \eqref{eq1:largexexpansion} to determine coefficients in terms of the $a_i$ and $b_i$. 

Equation \eqref{eq2:lineart} may be rewritten in the
context of connection preserving deformations to read
\begin{equation}\label{eq4:altdifft}
Y_n(x,qt) = (I - t(1-q)B_n(x,t))Y_n = \tilde{B}_n(x,t) Y_n(x,t).
\end{equation}
We use this relation and \eqref{eq4:orthconnection} to deduce
\begin{align*}
P_n(x,q t) =& (Y_{\infty,n}(x,q t))^{-1} Y_{0,n}(x,qt),\\
 =& (Y_{\infty,n}(x,t))\tilde{B}_n(x,t)^{-1} \tilde{B}_n(x,t) Y_{0,n}(x,t),\\
 =& P_n(x,t),
\end{align*}
which shows us that the connection is preserved under deformations in $t$.

Since $Y_n$ satisfies \eqref{eq1:linearn} and \eqref{eq2:lineart} we have a compatibility condition, which follows from a consideration of $D_{q,t}Y_{n+1}$,
\begin{equation}\label{eq3:comptn}
M_n(x,qt)B_n(x,t) + D_{q,t}M_n(x,t) = B_{n+1}(x,t)M_n(x,t).
\end{equation}
The first row of \eqref{eq3:comptn} is equivalent to 
\begin{subequations}\label{eq3:Fruedt}
\begin{gather}
 \frac{x-b_n}{a_{n+1}} \left[R-(1-q)t(S+\Psi_{n+1})\right] + (1-q)ta_{n+1}\Phi_{n+1} = \\
 \hspace{1cm} \frac{x-\hat{b}_n}{\hat{a}_{n+1}}\left[R-(1-q)t(S+\Psi_n)\right] +       
 \frac{(1-q)ta_n\hat{a}_n\Phi_{n-1}}{\hat{a}_{n+1}}\nonumber, \\
 \frac{a_n}{\hat{a}_{n+1}} \left[(1-q)t(x-\hat{b}_n)\Phi_n\right] + \frac{a_n}{a_{n+1}} \left[R-(1-q)t(S+\Psi_{n+1})\right]=  \\
 \hspace{1cm} \frac{\hat{a}_n}{\hat{a}_{n+1}}\left[R-(1-q)t(S + \Psi_{n-1} - (x-b_{n-1})\Phi_{n-1})\right] \nonumber.
\end{gather}
\end{subequations}
We have an additional relation 
\begin{equation}\label{eq3:detBntilde}
\det(I - t(1-q)B_n) = \frac{\det \hat{Y}_n}{\det Y_n} = \frac{a_n w}{\hat{a}_n \hat{w}} = \frac{a_nR}{\hat{a}_n(R - 2t(1-q)S)}.
\end{equation}
A consequence of this relation is the first order recurrence relation in $n$, given by
\begin{gather}
 \hat{a}_n \left((q-1) t \left(\Psi _n+S\right)+R\right) \left((q-1)t   
 \left(\Phi_{n-1}\left(b_{n-1}-x\right)+\Psi_{n-1}+S\right)+R\right)\\ 
  +(q-1)^2 t^2 a_n^2 \Phi _{n-1} \Phi _n \hat{a}_n-R a_n (2 (q-1) S t+R)= 0 . \nonumber 
\end{gather}

The compatibility condition \eqref{eq3:compxn} is naturally paired with \eqref{eq3:comptn}. Thus, rewriting these read
\begin{subequations}\label{eq3:compxnttilde}
\begin{align}
\label{eq3:compxntilde} \tilde{A}_{n+1}(x,t) M_n(x,t) =& M_n(qx,t)\tilde{A}_n(x,t),\\
\label{eq3:comptntilde} \tilde{B}_{n+1}(x,t) M_n(x,t) =& M_n(x,qt)\tilde{B}_n(x,t).
\end{align}
A further identity which can be grouped with these follows from the compatibility imposed by the requirement that $D_{q,t}D_{q,x} Y_n = D_{q,x} D_{q,t} Y_n$. One computes
\begin{equation}
\label{eq3:compxttilde} \tilde{A}_n(x,qt) \tilde{B}_n(x,t) = \tilde{B}_n(qx,t)\tilde{A}_n(x,t),
\end{equation}
\end{subequations}
which is equivalent to (\ref{eq3:Lax}).

\section{Deformed little $q$-Jacobi Polynomials}

The little $q$-Jacobi polynomials were introduced by Hahn \cite{Hahn}. This family of polynomials possesses the orthogonality relation  \cite{askey}
\[
\int_0^1 \frac{x^{\sigma}(qxb;q)_{\infty}}{(qx;q)_{\infty}} p_i(x) p_j(x)\mathrm{d}_qx = \delta_{ij}.
\]
This ratio of two exponential factors may be scaled and chosen appropriately so that the root and pole 
is at $a_3$ and $a_4$ respectively. We now adjoin two roots that are proportional to $t$, $a_1t$ 
and $a_2t$, to give the deformed weight
\begin{equation}\label{eq5:weight}
w(x,t) = \frac{x^\sigma\left( \frac{x}{a_1t}, \frac{x}{a_3} ;q\right)_\infty}{\left( \frac{x}{ta_2}, \frac{x}{a_4} ;q\right)_\infty}.
\end{equation}
In keeping with the notation of \cite{Sakai:qP6}, we trust there is no ambiguity 
between the terms in the three term recursion relation, $a_n$, and the roots of the 
determinant, $a_i$. The deformed polynomials associated with \eqref{eq5:weight} satisfy
\[
L(p_ip_j) = \int_S w(x,t) p_i(x,t)p_j(x,t)d_q x = \delta_{ij}.
\]
The set $S$, also called the support of the weight, may begin and end at distinct roots of $w(x,t)$. These include $a_3$, $a_1t$ and $0$. 
Choosing $a_3$ and $a_1t$ and using \eqref{eq2:Hypergeometricintegral} allows the
moments to be expressed in terms of Heine's basic hypergeometric function,
\begin{align*}
\mu_k =& \int_{qa_3}^{qa_1t} \frac{x^{\sigma+k} \left( \frac{x}{a_1t}, \frac{x}{a_3} ;q\right)_\infty}{\left( \frac{x}{ta_2}, \frac{x}{a_4} ;q\right)_\infty} d_q x, \\
 =& \frac{(qa_1t)^{\sigma+k+1} (1-q)\left( \frac{a_1q^{\sigma+k+2}}{a_2},q;q\right)_\infty}{\left(q^{\sigma+k+1},\frac{qa_1}{a_2};q\right)_{\infty}} {}_2\phi_1 \left(\begin{array}{c |} \frac{a_4}{a_3} , q^{\sigma+k+1} \\ \frac{a_1q^{\sigma+k+2}}{a_2} \end{array} \hspace{.1cm} q; \frac{qa_1t}{a_4} \right) \\
  &+ \frac{(qa_3)^{\sigma+k+1} (1-q)\left( \frac{a_3q^{\sigma+k+2}}{a_4},q;q\right)_\infty}{\left(q^{\sigma+k+1},\frac{qa_3}{a_4};q\right)_{\infty}} {}_2\phi_1 \left(\begin{array}{c |} \frac{a_2}{a_1} , q^{\sigma+k+1} \\ \frac{a_3q^{\sigma+k+2}}{a_4} \end{array} \hspace{.1cm} q; \frac{qa_3}{a_2t} \right). 
\end{align*}
This allows us to use \eqref{eq1:anandbndets} to express $a_n$ and $b_n$ in terms of 
determinants of basic hypergeometric functions.

This weight (\ref{eq5:weight}) satisfies the equation
\[
D_{q,x} w(x,t) = \left(\frac{a_2a_4 (x- a_1t)(x- a_3) - q^\sigma a_1a_3(x- a_2t)(x- a_4)}{a_2a_4 (x- a_1t)(x- a_3)x(1-q)}\right)w(x,t).
\]
A comparison with \eqref{eq2:qdifofweight} reveals that the spectral data polynomials are
\begin{align*}
W =& a_2a_4(x- a_1t)(x- a_3)x(1-q),\\
2V =& a_2a_4(x- a_1t)(x- a_3) - q^\sigma a_1a_3(x- a_2t)(x- a_4).
\end{align*}
Recalling Theorem \ref{thm:Dqxpn}, it follows that the poles of the linear $q$-difference equation in $x$ satisfied by 
these polynomials is determined by the polynomial
\begin{equation}
W-2x(1-q)V =  q^\sigma a_1a_3(x- a_2t)(x- a_4)x(1-q).
\end{equation}
In the $t$ direction, $w$ satisfies the equation
\[
D_{q,t}w(x,t) =\left( \frac{a_1(x-qa_2t) - a_2(x-qa_1t)}{a_1(x-qa_2t)t(1-q)}\right) w(x,t).
\]
Comparing this expression with \eqref{eq3:weightrationaltderive} shows
\begin{align*}
R(x,t) =& a_1(x-qa_2t)t(1-q),\\
2S(x,t) =& a_1(x-qa_2t) - a_2(x-qa_1t).
\end{align*}
The appropriate poles of the linear $q$-difference equation in $t$ satisfied by these polynomials 
is therefore determined by the polynomial
\begin{equation}
R- 2t(1-q)S = t(1-q)(x-qa_1t).
\end{equation}
We remark these explicit forms for $W,V,R$ and $S$ satisfy \eqref{eq3:compWVRS} as they must. 

\subsection{Linear problem}

Since we have an upper bound for $\deg_x\Theta_n, \deg_x\Omega_n, \deg_x\Phi_n$ and $\deg_x\Psi_n$ from \eqref{eq3:degthetaomega} and \eqref{eq3:degPhiPsi}, we parameterize $\Theta_n$, $\Omega_n$, $\Phi_n$ and $\Psi_n$ by
\begin{subequations}\label{eq5:omthphipsdeg}
\begin{align}
\Theta_n &= \theta_{0,n} + \theta_{1,n} x, \\
\Omega_n &= \omega_{0,n} + \omega_{1,n} x + \omega_{2,n}x^2, \\
\Phi_n &= \phi_{0,n}, \\
\Psi_n &= \psi_{0,n} + \psi_{1,n} x.
\end{align}
\end{subequations}
This bounds the degree of the relevant polynomial component of the linear $q$-difference equations in $x$ and $t$. Hence the linear $q$-difference equations satisfied by the polynomials may be written in the form \eqref{eq4:altdiffx} and \eqref{eq4:altdifft} where
\begin{align}
\tilde{A}_n = I - x(1-q)A_n &= \frac{\tilde{A}_{0,n} + \tilde{A}_{1,n}x + \tilde{A}_{2,n} x^2}{(x- a_2t)(x- a_4)},\\
\tilde{B}_n = I - t(1-q)B_n &= \frac{\tilde{B}_{0,n} + \tilde{B}_{1,n}x}{(x-qa_1t)},
\end{align}
for some set of $\tilde{A}_{i,n}$ and $\tilde{B}_{i,n}$. According to \eqref{eq3:detAntilde} and \eqref{eq3:detBntilde}, the determinants of these matrices are
\begin{align*}
\det \tilde{A}_n &= \frac{a_2a_4(x- a_1t)(x- a_3)}{a_1a_3q^\sigma (x- a_2t)(x- a_4)},\\
\det \tilde{B}_n &= \frac{a_1 a_n \left(x-a_2 q t\right)}{a_2 \left(x-a_1 q t\right) \hat{a}_n}.
\end{align*}

At this point, the associated linear $q$-difference equation satisfied by the orthogonal polynomials is one in which the coefficient matrix, $\tilde{A}_n$, is rational rather than polynomial. 
To relate this formulation to the classical theory of Birkhoff \cite{Birkhoff:General}, or more precisely, Jimbo and Sakai \cite{Sakai:qP6}, we require a gauge transformation that will relate the linear $q$-difference equation in which the coefficient matrix is rational to a linear $q$-difference equation in which the coefficient matrix is polynomial. By considering the associated $q$-difference equation for $Y_n^* = Z_n Y_n$, we note that $Y_n^*$ satisfies the trio of equations
\begin{subequations}\label{eq5:ynstarevol}
\begin{align}
\label{eq5:ynstarevolx}Y_n^*(qx,t) &= \left( \overline{Z}_n(I - x(1-q)A_n) Z_n^{-1} \right) Y_n^* = A_n^* Y_n^*,\\
\label{eq5:ynstarevolt}Y_n^*(x,qt) &= \left( \hat{Z}_n (I - t(1-q)B_n) Z_n^{-1} \right) Y_n^* = B_n^* Y_n^*,\\
\label{eq5:ynstarevoln}Y_{n+1}^* (x,t) &= \left( Z_{n+1} M_n Z_n^{-1} \right) Y_n^* = M_n^* Y_n^*.
\end{align}
\end{subequations}
By letting $Z_n$ to be proportional to appropriate $q$-exponential factors allows 
$A_n^*$ to be polynomial. We may also choose $Z_n$ carefully so that $A_n^*$ possesses some desirable properties, such as certain asymptotic characteristics in $x$ and/or $t$, and doing so makes
the correspondence to the work of Jimbo and Sakai \cite{Sakai:qP6} more apparent. Specifically,
by choosing 
\[
Z_n(x,t) = \frac{e_{q,a_1a_2a_3a_4tq^\sigma}(x)}{\left( \frac{x}{a_2t},\frac{x}{a_4}; q\right)_\infty } \begin{pmatrix} {\displaystyle \frac{a_2a_4q^{-n}}{e_{q,qa_2}(t)\gamma_n}}  & 0 \\ 0 &{\displaystyle \frac{\gamma_{n-1}}{e_{q,qa_1}(t)}} \end{pmatrix},
\]
we have that $Y_n^*$ satisfies the $q$-difference equations 
\begin{subequations}
\begin{align}
\label{eq5:linearAstar} Y_n^*(qx,t) =& (A_{0,n}^* + A_{1,n}^*x + A_{2,n}^* x^2 )Y_n^*  = A_n^* Y_n^*, \\
\label{eq5:linearBstar} Y_n^*(x,qt) =& \frac{x(B_{0,n}^* + B_{1,n}^*x)}{(x-qa_1t)(x-qa_2t)}Y_n^* = B_n^* Y_n^*.
\end{align}
\end{subequations}
The corresponding determinants are given by \eqref{eq3:detAntilde}, \eqref{eq3:detBntilde} and \eqref{eq4:eprops}
\begin{subequations}\label{eq5:determinants}
\begin{align}
\label{eq5:detAnstar}\det(A_n^*) =& a_1a_2a_3a_4q^\sigma(x-a_1t)(x-a_2t)(x-a_3)(x-a_4),\\
\label{eq5:detBnstar}\det(B_n^*) =& \frac{t^2x^2}{ (x-q t a_1)(x-qta_2)}.
\end{align}
\end{subequations}
It will transpire that the form of the coefficient matrices of \eqref{eq5:linearAstar} and \eqref{eq5:linearBstar} is well suited for the purpose of
parameterizing the linear problem satisfied by the orthogonal polynomials.

Although \eqref{eq5:linearAstar} specifies a $2\times 2$ linear $q$-difference system in which the determinant of coefficient matrix, given by \eqref{eq5:detAnstar}, has roots that coincide with those found in \cite{Sakai:qP6}, we require two additional properties; firstly that $A_{2,n}^*$ is a constant diagonal matrix and secondly, that $A_{0,n}^*$ is semisimple with eigenvalues proportional to $t$. 

An asymptotic expansion of $\Omega_n$ and $\Theta_n$ around $x=\infty$ reveals
\begin{align*}
\omega_{2,n} =& \frac{a_2a_4 - q^\sigma(2q^n - 1)a_1a_3}{2}, \\
\theta_{1,n} =& \frac{a_2a_4}{q^{n+1}} - q^{n+\sigma} a_1a_3, \\
\end{align*}
giving
\begin{equation}\label{eq5:A2nstar}
A_{2,n}^* = \begin{pmatrix} \kappa_1 & 0 \\ 0 & \kappa_2  \end{pmatrix},
\end{equation}
where 
\begin{equation*}
\kappa_1 = q^{n+\sigma}a_1 a_3,\qquad
\kappa_2 = a_2a_4q^{-n}.
\end{equation*}
This shows that the linear problem possesses the first required property. 

To show that $A_{0,n}^*$ has eigenvalues that are proportional to $t$, we first write $A_{0,n}^*$ as
\begin{align}\label{eq5:A0nstar1}
A_{0,n}^* = \begin{pmatrix}  {\displaystyle \frac{\kappa_1\kappa_2t+ a_1a_2a_3a_4t}{2}  -\omega_{0,n}}& {\displaystyle \frac{q \kappa_2 w_n\theta_{0,n}}{\kappa_2-q\kappa_1} }\\
{\displaystyle \frac{a_n^2( q\kappa_1 - \kappa_2)\theta_{0,n-1}}{q\kappa_2w_n} } & {\displaystyle \frac{t\kappa_1\kappa_2+ a_1a_2a_3a_4t}{2} - b_{n-1}\theta_{0,n-1} - \omega_{0,n-1}}
\end{pmatrix},
\end{align}
where we have used the notation
\begin{equation}\label{eq5:wndef}
w_n = \frac{ e_{q,qa_1}(t)(\kappa_2- q\kappa_1)}{qe_{q,qa_2}(t)\gamma_n^2}.
\end{equation}
We let $\lambda_{1,n}$ and $\lambda_{2,n}$ be the eigenvalues of $A_{0,n}^*$. Utilizing \eqref{eq5:determinants} and that $\det(A_n^*(0,t)) = \lambda_{1,n} \lambda_{2,n}$ gives us
\[
\lambda_{1,n} \lambda_{2,n} = \kappa_1\kappa_2 a_1a_2a_3a_4t^2,
\]
revealing that $\hat{\lambda}_{1,n} \hat{\lambda}_{2,n} = q^2\lambda_{1,n} \lambda_{2,n}$. Adding $B_{0,n}^*$ to both sides of the residue of \eqref{eq3:compxttilde} at $x = 0$, namely the relation
\[
\hat{A}_{0,n}^* B_{0,n}^* = q B_{0,n}^* A_{0,n}^*,
\]
and then taking determinants shows
\[
\det(I+\hat{A}_{0,n}^*) = \det(I+q A_{0,n}^*),
\]
revealing
\[
1 + \hat{\lambda}_{1,n} + \hat{\lambda}_{2,n} + \hat{\lambda}_{1,n}\hat{\lambda}_{2,n} = 1 + q \lambda_{1,n} + q \lambda_{2,n} + q^2 \lambda_{1,n} \lambda_{2,n}.
\]
This shows that $\lambda_{1,n} \lambda_{2,n} \propto t^2$ and $\lambda_{1,n} + \lambda_{2,n} \propto t$, hence $\lambda_{1,n}$ and $\lambda_{2,n}$ are proportional to $t$. 

A further property of $\lambda_{1,n}$ and $\lambda_{2,n}$, that will be useful later on, is their independence of $n$. The independence of
$\kappa_1\kappa_2$'s on $n$ indicates that $\lambda_{1,n}\lambda_{2,n}$ is independent of $n$. However the trace of $A_{0,n}^*$ is
\[
\kappa_1\kappa_2t\left(1+ \frac{1}{q^{\sigma}}\right) - (\omega_{0,n} + \omega_{0,n-1} + b_{n-1}\theta_{0,n-1}) = \lambda_{1,n} + \lambda_{2,n},
\]
which indicates that the constant coefficient of \eqref{eq3:Fruedx} may be expressed in terms of $\lambda_{1,n}$ and $\lambda_{2,n}$ as
\[
\lambda_{1,n} + \lambda_{2,n} - \lambda_{1,n+1} - \lambda_{2,n+1} = 0.
\]
This proves $\lambda_{1,n} + \lambda_{2,n}$ is independent of $n$, hence $\lambda_{1,n}$ and $\lambda_{2,n}$ are independent of $n$. We may now write 
\begin{align}\label{eq5:eigenvalues}
\{\lambda_{1,n},\lambda_{2,n}\} = \{ \theta_1 t,\theta_2 t \},
\end{align} 
where $\theta_1$ and $\theta_2$ are constant in $t$ and $n$. These eigenvalues are not free, with an implicit dependence on the $a_i$'s and $\kappa_i$'s and the support chosen.

The additional properties mean that \eqref{eq5:linearAstar} can be cast in a form equivalent to the linear problem studied in \cite{Sakai:qP6}. Technical achievements in \cite{Sakai:qP6} are to identify the parameterization of the linear problem which leads to the $q$-$\mathrm{P}_{\mathrm{VI}}$ system. The present orthogonal polynomial setting allows us to perform these steps in a more detailed, and perhaps more systematic manner.

\subsection{Orthogonal polynomial parameterization}

Our pathway toward the parameterization of the problem is to make use of the orthogonal polynomial variables. Parameterizations of this sort can be found in previous works \cite{B03, B94, F08, F98}. However, these works do not provide a systematic way to link up with co-ordinates that specify Painlev\'e systems. 

To begin, using the expansions  \eqref{eq2:defThetaOmega} of $\Omega_n$ and $\Theta_n$, we find 
\begin{subequations}\label{eq5:coefficientthetaomega}
\begin{align}
\omega_{1,n} =& \frac{(1-q) \kappa _1}{q}\sum _{i=0}^{n-1} b_i+\kappa _1 \left(a_2 t+a_4\right)-\frac{\kappa _1 \kappa _2 \left(a_2 t+a_4\right)}{2 a_2 a_4}-\frac{1}{2} a_2
   a_4 \left(a_1 t+a_3\right),\\
\label{eq5:theta0n} \theta_{0,n} =& \frac{\kappa _1 \left(a_2 q t+a_4 q- q \sum _{i=0}^n b_i+\sum _{i=0}^{n-1} b_i\right)}{q} -\frac{\kappa _2 \left(a_1 q t+a_3 q + q \sum _{i=0}^{n-1} b_i-\sum _{i=0}^n
   b_i\right)}{q^2},\\
\label{eq5:omega0n} \omega_{0,n} =&  a_2 t \kappa _1 \left(\sum _{i=0}^{n-1} b_i\right)+  a_4 \kappa _1 \left(\sum _{i=0}^{n-1} b_i\right)- \frac{\kappa _1a_2 t}{q}  \left(\sum _{i=0}^{n-1}
   b_i\right)- \frac{a_4 \kappa _1}{q} \left(\sum _{i=0}^{n-1} b_i\right) \nonumber \\ &  - \kappa _1 \left(\sum _{i=1}^n a_i^2\right)+ \frac{ \kappa _1}{q^2} \sum _{i=1}^{n-1} a_i^2
   - \frac{\kappa_1}{q^2} \sum _{i=0}^{n-2} \sum _{j=i+1}^{n-1} b_i b_j - \frac{\kappa _1}{q^2} \sum _{i=0}^{n-1} \sum _{j=i}^{n-1} b_i b_j \nonumber \\ 
  & + \frac{\kappa_1}{q} \left(\sum _{i=0}^{n-1} b_i\right)^2+
 \frac{\kappa_2}{q} a_n^2- a_2 a_4 t \kappa _1+\frac{a_1 a_2 a_3 a_4 t}{2} +\frac{ t \kappa _1 \kappa _2}{2}.
\end{align}
\end{subequations}
The expansions \eqref{eq2:defPhiPsi} of $\Phi_n$ and $\Psi_n$ give 
\begin{subequations}\label{eq5:phisivals}
\begin{align}
\psi_{1,n} =& \frac{1}{2} \left(a_1+a_2\right) -\frac{a_2  \hat{\gamma}_n}{\gamma_n},\\
\psi_{0,n} =& \frac{a_2 \hat{\gamma}_n  \left(\sum _{i=0}^{n-1} \hat{b}_i-\sum _{i=0}^{n-1} b_i+a_1 q t\right)}{\gamma_n}-a_1 a_2 q t,\\
\phi_{0,n} =& \frac{a_1 \gamma_n^2-a_2 \hat{\gamma}_n^2}{\gamma_n \hat{\gamma}_n }.
\end{align}
\end{subequations}
This specifies a parameterization of the linear problem for $Y_n^*$ in terms of orthogonal polynomial variables. We use the notation
\begin{equation}\label{eq5:Gammadef}
\Gamma_n = \sum_{i=0}^{n-1} b_i,
\end{equation}
which is proportional to the coefficient of $x^{n-1}$ in $p_n$. By combining \eqref{eq5:coefficientthetaomega}, \eqref{eq5:omthphipsdeg}, \eqref{eq4:Aform} and \eqref{eq5:linearAstar}, 
\begin{equation}\label{eq5:A1nstar}
A_{1,n}^* = \left(
\begin{array}{cc}
 {\displaystyle \frac{(q-1) \kappa _1\Gamma_n }{q} -\kappa _1 \left(a_2 t+a_4\right) }& \kappa _2 w_n \\
 {\displaystyle \frac{a_n^2 \left(q \kappa _1-\kappa _2\right) \left(q \kappa _2-\kappa _1\right)}{q^2 \kappa _2 w_n} } & 
 {\displaystyle -\frac{(q-1) \kappa _2\Gamma_n }{q}-\kappa _2 \left(a_1 t+a_3 \right)}
\end{array}
\right).
\end{equation}
We make use of the relations 
\begin{align*}
\mathrm{trace} A_{0,n}^* =& \theta_1t + \theta_2t,\\
\det A_{0,n}^* =& \theta_1 \theta_2t^2,
\end{align*}
which allows \eqref{eq5:A0nstar1} to be simplified to 
\begin{equation}\label{eq5:preA0nstar}
A_{0,n}^* = \left(
\begin{array}{cc}
 {\displaystyle \frac{t\left(a_1 a_2 a_3 a_4 + \kappa _1 \kappa _2\right)}{2}}-\omega_{0,n}  & {\displaystyle -\frac{q \kappa _2 w_n \theta _{0,n}}{q \kappa _1-\kappa_2}}  \\
{ \displaystyle \frac{\left(q \kappa _1-\kappa _2\right) a_n^2 \theta _{0,n-1}}{q \kappa _2 w_n} }& \omega _{0,n}-{\displaystyle \frac{ t \left(a_1 a_2 a_3 a_4-2 \left(\theta_1+\theta_2\right)+\kappa _1 \kappa _2\right)}{2}}
\end{array}
\right),
\end{equation}
where 
\[
\theta_{0,n-1} = \frac{\left(a_1 a_2 a_3 a_4 t-2 t \theta _1+t \kappa _1 \kappa_2 -2 \omega _{0,n} \right) \left(a_1 a_2 a_3 a_4 t-2 t \theta _2+t \kappa _1 \kappa_2-2 \omega _{0,n} \right)}{4 a_n^2 \theta _{0,n}}.
\]
We simplify the expression for $A_{0,n}^*$ by introducing the variable $r_n$ so that we may write $A_{0,n}^*$ as
\begin{equation}\label{eq5:A0nstar}
A_{0,n}^* = \left(
\begin{array}{cc}
 \theta _1t+r_n &  {\displaystyle -\frac{q \kappa _2 w_n \theta _{0,n}}{q \kappa _1-\kappa _2} } \\
 {\displaystyle \frac{\left(q \kappa _1-\kappa _2\right) r_n \left(r_n+  \theta _1t- \theta _2t\right)}{q \kappa _2 w_n \theta_{0,n}} } &  \theta_2 t - r_n
\end{array}
\right).
\end{equation}
In relating the coefficient of $x^2$ in $\det A_n^*$ with \eqref{eq5:detAnstar}, we express $r_n$ as 
\begin{multline*}
r_n = -\frac{(q-1) \kappa _1 \kappa _2 \left(a_1 t-a_2 t+a_3-a_4\right)\Gamma_n }{q \left(\kappa _1-\kappa _2\right)}-\frac{(1-q)^2 \kappa _1 \kappa _2
  \Gamma_n^2}{q^2 \left(\kappa _1-\kappa _2\right)}\\
    -\frac{a_n^2 \left(q \kappa _1-\kappa _2\right) \left(q \kappa _2-\kappa
   _1\right)}{q^2 \left(\kappa _1-\kappa _2\right)}+\frac{t \left(\theta _2 \kappa _1+\theta _1 \kappa
   _2-a_1 a_3 \kappa _2 \kappa _1-a_2 a_4 \kappa _2 \kappa _1\right)}{\kappa _1-\kappa _2}.
\end{multline*}
Equating \eqref{eq5:A0nstar} with \eqref{eq5:preA0nstar} gives an alternate representation of $\omega_{0,n}$ to that of \eqref{eq5:omega0n}. The equations \eqref{eq5:A0nstar}, \eqref{eq5:A1nstar} and \eqref{eq5:A2nstar} are explicit parameterizations of the linear problem using orthogonal polynomial variables combined with knowledge of the structures \eqref{eq5:linearAstar} and \eqref{eq4:Aform}.

We now turn our attention to the parameterization of the linear problem, \eqref{eq5:linearBstar}, involving $B_n^*$. First, upon recalling \eqref{eq4:Bform}, it follows from the large $x$ expansion of $\Phi_n$ and $\Psi_n$, as implied by \eqref{eq5:phisivals}, that
\[
B_{1,n}^* = -t I.
\]
Direct substitution of the values of $\Psi_n$ and $\Phi_n$ from \eqref{eq5:phisivals} gives
\begin{equation}\label{eq5:B0nstar}
B_{0,n}^* =  t\left(
\begin{array}{cc}
  \hat{\Gamma}_n-\Gamma_n+a_1 q t & {\displaystyle \frac{q  \kappa _2 \left(\hat{w}_n-w_n\right)}{q \kappa _1-\kappa _2} } \\
 {\displaystyle \frac{ \left(q \kappa _1-\kappa _2\right) \left(w_n \hat{a}_n^2-a_n^2 \hat{w}_n\right)}{q \kappa _2 w_n \hat{w}_n}} & {\displaystyle \Gamma_n- \hat{\Gamma}_n+ q a_2 t}
\end{array}
\right).
\end{equation}
This gives us enough information to deduce the evolution of the variables $\gamma_n^2$, $a_n^2$ and $\Gamma_n$, which completes the parameterization of the linear problem in terms of the orthogonal polynomial variables. 

Use will be made of \eqref{eq5:A2nstar}, \eqref{eq5:A1nstar}, \eqref{eq5:A0nstar} and \eqref{eq5:B0nstar} as we now proceed to make the correspondence between the above discrete dynamical system and $q$-$\mathrm{P}_{\mathrm{VI}}$ by making a correspondence between the linear systems. 
\subsection{Jimbo-Sakai parameterization}
Our primary task is to find expressions for $w_n$, $y_n$ and $z_n$ in terms of $\gamma_n^2$, $a_n^2$ and $b_n$ and vise versa. We have chosen $w_n$ in the previous parameterization, as it is related to $\gamma_n^2$ via \eqref{eq5:wndef}, to be the variable that reflects the gauge freedom in both parameterizations of the linear problem. In keeping with earlier remarks, we deduce
\begin{equation}\label{eq5:yztheta}
\theta_{0,n}  = \frac{y_n(q\kappa_1 - \kappa_2)}{q},
\end{equation}
and define variables $z_1$ and $z_2$ according to
\begin{equation}\label{eq5:Aynform}
A_n^*(y_n,t) = \begin{pmatrix} \kappa_1 z_{1} & 0 \\ \ast & \kappa_2z_{2} \end{pmatrix}.
\end{equation}
Evaluating the determinant at $x=y_n$ reveals
\begin{equation}\label{eq5:zfactors}
z_{1}z_{2} = (y_n- a_1t)(y_n-a_2t)(y_n-a_3)(y_n-a_4).
\end{equation}
We factorize this into the factors
\begin{subequations}\label{eq3:zvals}
\begin{align}
\label{eq3:zval1}z_{1} =& \frac{(y_n-ta_1)(y_n-ta_2)}{q\kappa_1z_n},\\
\label{eq3:zval2}z_{2} =& (y_n-a_3)(y_n-a_4)q\kappa_1z_n.
\end{align}
\end{subequations}
The benefit of this particular factorization reveals itself in the proof of Theorem \ref{thm:qP6}. It follows from \eqref{eq5:A2nstar}, \eqref{eq5:A0nstar1} and \eqref{eq5:A0nstar} that $z_{1}$ and $z_{2}$ may be expressed in terms of $a_n^2$ and $\Gamma_n$ via the expressions 
\begin{subequations}
\begin{align}
\label{eq4:z1n}\kappa_1 z_{1} =& -\frac{\kappa _1 \theta _{0,n} \left((q-1) \Gamma_n -q \left(a_2 t+a_4\right)\right)}{\kappa _2-q \kappa _1}+\frac{q^2 \kappa _1 \theta_{0,n}^2}{\left(\kappa _2-q \kappa
   _1\right){}^2}+r_n +t \theta _1, \\
\label{eq4:z2n} \kappa_2z_{2} =& \frac{\kappa _2 \theta _{0,n} \left((q-1) \Gamma_n +a_1 q t+a_3 q\right)}{\kappa _2-q \kappa _1}+\frac{q^2 \kappa _2 \theta _{0,n}^2}{\left(\kappa _2-q \kappa
   _1\right)^2}-r_n+t \theta _2,
\end{align}
\end{subequations}
which specifies $z_n$. To be consistent with \eqref{eq5:Aynform}, the matrix in \eqref{eq5:linearAstar} permits the parameterization \cite{Sakai:qP6}
\begin{equation}\label{eq5:AS}
A_n^* = \begin{pmatrix} 
\kappa_1((x-y_n)(x-\alpha) + z_{1}) & \kappa_2 w_n (x-y_n) \\
{\displaystyle \frac{\kappa_1 (\gamma x + \delta)}{w_n}} & \kappa_2 ((x-y_n)(x-\beta)+z_{2})
\end{pmatrix},
\end{equation}
where $\alpha$, $\beta$, $\gamma$ and $\delta$ are to be determined. Comparing the upper left entry of \eqref{eq5:AS} with \eqref{eq5:A2nstar}, \eqref{eq5:A1nstar} and \eqref{eq5:A0nstar} shows
\begin{align}
\label{eq4:sumbnalpha} \Gamma_n =&\frac{q \left(a_2 t+a_4-y_n-\alpha \right)}{q-1},\\
r_n =& \kappa _1 y_n \alpha +\kappa _1 z_{1}-t \theta _1\nonumber.
\end{align}
These substituted into \eqref{eq4:z2n} reveal
\begin{equation}\label{eq3:alphaval}
\alpha = \frac{1}{\kappa_1- \kappa_2} \left(\frac{1}{y_n} (( \theta_1 + \theta_2)t - \kappa_1 z_{1} - \kappa_2z_{2}) - \kappa_2((a_1+a_2)t + a_3 + a_4- 2y_n)  \right).
\end{equation}
Conversely, comparing coefficients of the lower-left entry of \eqref{eq5:AS} with \eqref{eq5:A2nstar}, \eqref{eq5:A1nstar} and \eqref{eq5:A0nstar} gives
\begin{align}
\label{eq4:sumbnbeta} \Gamma_n =&-\frac{q \left(a_1 t+a_3-y_n-\beta\right)}{q-1},\\
r_n =& -\kappa _2 y_n \beta -\kappa _2 z_{2}+t \theta _2\nonumber.
\end{align}
These substituted into \eqref{eq4:z1n} show
\begin{equation}\label{eq3:betaval} 
\beta = \frac{1}{\kappa_1- \kappa_2} \left(- \frac{1}{y_n} (( \theta_1 + \theta_2)t - \kappa_1 z_{1} - \kappa_2z_{2}) + \kappa_1((a_1+a_2)t + a_3 + a_4- 2y_n)  \right).
\end{equation}
The strategy to be used to specify $\gamma$ and $\delta$ makes use of \eqref{eq5:detAnstar}. By equating the coefficient of $x^2$ of $\det A_n^*$ from \eqref{eq5:AS} with \eqref{eq5:detAnstar}, we have
\begin{equation}\label{eq3:gammaval}
\gamma = z_{1} + z_{2} + (y_n+\alpha)(y_n+\beta) + (\alpha+\beta)y_n - a_1a_2t^2 - (a_1+a_2)(a_3+a_4)t-a_3a_4.
\end{equation}
A comparison between the coefficient of $x$ in \eqref{eq5:AS} with that of \eqref{eq5:detAnstar} shows
\begin{equation}\label{eq3:deltaval}
\delta = \frac{1}{y_n}\left(a_1a_2a_3a_4t^2 - (\alpha y_n+z_{1})(\beta y_n + z_{2})\right).
\end{equation}
This concludes our task of parameterizing the linear problem associated with the orthogonal polynomial system with the weight \eqref{eq5:weight} and its correspondence with the parameterization of \cite{Sakai:qP6}. 
After completing the task of parameterization of $\alpha$, $\beta$, $\gamma$ and $\delta$, Jimbo and Sakai proceeded to give the coupled equations referred to as $q$-$\mathrm{P}_{\mathrm{VI}}$. However few details were given there. We shall provide details by making use of a structured form of the $B$ matrix following from the orthogonal polynomial viewpoint. The structured form of the $B$ matrix follows by 
using the substitutions of \eqref{eq4:sumbnalpha} and \eqref{eq4:sumbnbeta} in \eqref{eq5:B0nstar}, giving 
\begin{equation}\label{eq4:B0matrix}
B_{0,n}^* = \begin{pmatrix} {\displaystyle qt^2 \left( a_1 + a_2 - D_{q,t}(y_n + \alpha) \right)} & -{\displaystyle \frac{qt\kappa_2(w_n - \hat{w}_n)}{q\kappa_1 - \kappa_2} } \\
{\displaystyle \frac{qt\kappa_1(\hat{w}_n\gamma - w_n \hat{\gamma})}{(\kappa_1-q\kappa_2)w_n\hat{w}_n}  } & {\displaystyle qt^2\left( a_1 + a_2 - D_{q,t}(y_n + \beta) \right)} \end{pmatrix}.
\end{equation}
In addition to (\ref{eq4:B0matrix}) a crucial ingredient in our derivation of $q$-$\mathrm{P}_{\mathrm{VI}}$ 
are the compatibility conditions \eqref{eq3:compxnttilde}. 
After making the transformation $Y_n^{*} = Z_n Y_n$ these latter conditions read 
\begin{subequations}\label{eq5:compxtn}
\begin{gather}
\label{eq5:compxt}A_n^*(x,qt) B_n^*(x,t) = B_n^*(q x,t) A_n^*(x,t),\\
\label{eq5:compxn}A_{n+1}^*(x,t) M_n^*(x,t) = M_n^*(q x,t) A_n^*(x,t),\\
\label{eq5:comptn}B_{n+1}^*(x,t) M_n^*(x,t) = M_n^*(x,q t) B_n^*(x,t).
\end{gather}
\end{subequations}
By evaluating the residue of \eqref{eq5:compxt} at $x = a_1t, a_2t, qa_1t,qa_2t$, we obtain the expressions
\begin{subequations}\label{eq5:CompRes1}
\begin{align}
\label{Res1}(qa_1t B_{1,n}^* + B_{0,n}^*)A_n^*(a_1t,t) =& 0, \\
\label{Res2}(qa_2t B_{1,n}^* + B_{0,n}^*)A_n^*(a_2t,t) =& 0, 
\end{align}
\end{subequations}
\begin{subequations}\label{eq5:CompRes2}
\begin{align}
\label{Res3} A_n^*(qa_1t,qt)(qa_1t B_{1,n}^* + B_{0,n}^*) =& 0, \\
\label{Res4} A_n^*(qa_2t,qt)(qa_2t B_{1,n}^* + B_{0,n}^*) =& 0 .
\end{align}
\end{subequations}
By looking at the residue of \eqref{eq5:compxt} at $x = 0$, we obtain the additional relation
\begin{equation}\label{eq5:CompRes0}
\hat{A}_{0,n}^* B_{0,n}^* = q B_{0,n}^* A_{0,n}^* .
\end{equation}
\begin{theorem}[\cite{Sakai:qP6}]\label{thm:qP6}
The compatibility condition, \eqref{eq5:compxt}, is equivalent to the evolution equations for $y_n$, $z_n$ and $w_n$ specified by
\begin{subequations}\label{eq5:qP6theta}
\begin{align}
\label{eq5:what}\hat{w}_n &= w_n \frac{(q\kappa_1\hat{z}_n-1)}{\kappa_2\hat{z}_n-1},\\
\label{eq5:zhat}\hat{z}_n z_n &= \frac{(y_n - a_1t)(y_n-ta_2)}{q\kappa_1\kappa_2(y_n-a_3)(y_n-a_4)},\\
\label{eq5:yhat}\hat{y}_n y_n &= \frac{q(\theta_1\hat{z}_n - ta_1a_2)(\theta_2\hat{z}_n - ta_1 a_2)}{a_1a_2(q\kappa_1\hat{z}_n - 1)(\kappa_2 \hat{z}_n - 1)}.
\end{align}
\end{subequations}
\end{theorem}

\begin{proof}
For brevity, we let the parameterization of $B_{0,n}^*$ of \eqref{eq4:B0matrix} be given by $B_{0,n}^* = (b_{ij})_{i,j=1,2}$. The upper right entries of compatibility condition \eqref{Res3} and \eqref{Res4} read
\begin{align*}
\kappa_2\hat{w}_n(\hat{y}_n- qta_1)(b_{22} - qa_1t^2) &= \kappa_1b_{12} ((qta_1-\hat{y}_n)(qta_1-\hat{\alpha})+\hat{z}_1),\\
\kappa_2\hat{w}_n(\hat{y}_n- qta_2)(b_{22} - qa_2t^2) &= \kappa_1b_{12} ((qta_2-\hat{y}_n)(qta_2-\hat{\alpha})+\hat{z}_1),
\end{align*}
which gives us an expression for $b_{12}$ and $b_{22}$. We only require $b_{12}$, which is given by
\begin{equation*}
b_{12} = \frac{t\kappa_2 \hat{w}_n (\hat{y}_n - qa_1t)(\hat{y}_n - qa_2t)}{\kappa_1((\hat{y}_n - qa_1t)(\hat{y}_n - qa_2t)- \hat{z}_{1})}.
\end{equation*}
Equating this with the upper right element of \eqref{eq4:B0matrix} gives
\[
-\frac{q(\hat{w}_n-w_n)}{q\kappa_1 - \kappa_2} = \frac{\hat{w}_n (\hat{y}_n - qa_1t)(\hat{y}_n - qa_2t)}{\kappa_1(\hat{z}_{1} - (\hat{y}_n - qa_1t)(\hat{y}_n - qa_2t))}.
\]
This evolution equation is simplified using the particular factorization \eqref{eq5:zfactors}. The structure of the right hand side of the above relation justifies, {\it a posteri}, the factorization \eqref{eq5:zfactors}. The particular form of \eqref{eq3:zvals} means the evolution  of $w_n$ is equivalent to \eqref{eq5:what}.

The upper right entries of compatibility condition \eqref{Res1} and \eqref{Res2} read
\begin{align*}
\kappa_2w_n(y_n-a_1t)(b_{11} - qa_1t^2) &= \kappa_2b_{12} ((ta_1-y_n)(ta_1-\beta)+z_2),\\
\kappa_2w_n(y_n-a_2t)(b_{11} - qa_2t^2) &= \kappa_2b_{12} ((ta_2-y_n)(ta_2-\beta)+z_2),
\end{align*}
which we solve in terms of $b_{11}$ and $b_{12}$ to give
\begin{align}
b_{12} =& \frac{qtw_n(y_n-a_1t)(y_n-a_2t)}{(y_n-a_1t)(y_n-a_2t)-z_2},\\
\label{b11} b_{11} =& \frac{q t \left(z_2 \left(y_n-\left(a_1+a_2\right) t\right)+\beta  \left(a_1 t-y_n\right) \left(a_2 t-y_n\right)\right)}{\left(a_1 t-y_n\right) \left(a_2 t-y_n\right)-z_2}.
\end{align}
We deduce
\[
\frac{\kappa_2(\hat{w}_n-w_n)}{q\kappa_1 - \kappa_2} = \frac{w_n(y_n-a_1t)(y_n-a_2t)}{(y_n-a_1t)(y_n-a_2t)-z_2},
\]
which is equivalent to \eqref{eq5:zhat} knowing \eqref{eq5:what}. Comparing \eqref{b11} with \eqref{eq4:B0matrix} yields 
\[
t \left( a_1 + a_2 + D_{q,t}(y_n + \alpha) \right) = \frac{z_2 \left(y_n-\left(a_1+a_2\right) t\right)+\beta  \left(a_1 t-y_n\right) \left(a_2 t-y_n\right)}{\left(a_1 t-y_n\right) \left(a_2 t-y_n\right)-z_2},
\]
which is equivalent to \eqref{eq5:yhat} knowing \eqref{eq5:what} and \eqref{eq5:zhat}, or the particular Riccati solutions
\[
\hat{y}_n = \frac{qy_n(1-\kappa_2\hat{z}_n)}{1-q\kappa_2\hat{z}_n},
\]
the latter not being satisfied in general. The derivation of the evolution equations is complete.
\end{proof}
Full correspondence with the Jimbo and Sakai form is obtained by letting
\begin{equation}
a_5 = \frac{a_1a_2}{\theta_1}, \qquad a_6 = \frac{a_1a_2}{\theta_2}, \qquad a_7 = \frac{1}{q\kappa_1}, \qquad a_8 = \frac{1}{\kappa_2},
\end{equation}
where \eqref{eq5:qP6theta} become
\begin{align*}
\hat{z}_n z_n =& \frac{a_7a_8(y_n - a_1t)(y_n-ta_2)}{(y_n-a_3)(y_n-a_4)},\\
\hat{y}_n y_n =& \frac{a_3a_4(\hat{z}_n - a_5t)(\hat{z}_n - a_6t)}{(\hat{z}_n-a_7)(\hat{z}_n - a_8)},
\end{align*}
under conditions that
\[
\frac{a_5a_6}{a_7a_8} = \frac{qa_1a_2}{a_3a_4},
\]
as given in \cite{Sakai:qP6}.

We now return to the orthogonal polynomial context for these results. In addition to \eqref{eq5:ynstarevolx} the three term recursion relation, \eqref{eq1:linearn}, in the orthogonal polynomial context gives us another linear problem. The representation of $M_n^*$ following from \eqref{eq1:linearn} and \eqref{eq5:ynstarevoln} is
\begin{equation}\label{eq5:Mnstar}
M_n^* = \left(
\begin{array}{cc}
{\displaystyle \frac{x-b_n}{q}} & {\displaystyle \frac{\kappa _2 w_n}{q \kappa _1-\kappa _2} }\\
 {\displaystyle \frac{\kappa _2-q \kappa _1}{q \kappa _2 w_n}}  & 0
\end{array}
\right).
\end{equation}
This can be used to express the orthogonal polynomial quantity $b_n$ in terms of the natural variables. 
Considering the coefficient of $x^2$ and $x$ in the upper left and right entries of \eqref{eq5:compxn} respectively results in the expression
\begin{equation}\label{eq5:bn}
b_n = \frac{q \left(q \kappa _1 \alpha-\kappa _2 \beta\right)}{q^2 \kappa _1-\kappa _2}.
\end{equation}
For the orthogonal polynomial quantity $a_n^2$ a comparison of the lower left component of $A_{1,n}^*$ given by \eqref{eq5:A1nstar} and \eqref{eq5:AS} shows 
\begin{equation}\label{eq3:antrans}
a_n^2 = \frac{q^2\kappa_1\kappa_2 \gamma}{(q\kappa_1-\kappa_2)(q\kappa_2-\kappa_1)}.
\end{equation}

One important consequence from this perspective is that the natural variables may be expressed in 
terms of determinants of the moments. Using \eqref{eq5:wndef} and \eqref{eq2:detgamma} we have
\begin{equation}\label{eq5:detwn}
w_n = \frac{ e_{q,qa_1}(t)(\kappa_2- q\kappa_1)\Delta_{n+1}}{qe_{q,qa_2}(t)\Delta_n}.
\end{equation}
Using \eqref{eq5:theta0n} and \eqref{eq5:yztheta} gives
\begin{equation}\label{eq5:detyn}
y_n = \frac{q\kappa_1 (a_2t+a_4) - \kappa_2(a_1t+a_3)}{q\kappa_1 - \kappa_2} + \frac{\kappa_1-\kappa_2}{q\kappa_1-\kappa_2} \frac{\Sigma_n}{\Delta_n} - \frac{q\kappa_1 - \frac{\kappa_2}{q}}{q\kappa_1 - \kappa_2} \frac{\Sigma_{n+1}}{\Delta_{n+1}}.
\end{equation}
The simplest determinantal form for $z_n$ comes from the substitution of \eqref{eq5:detwn} into the inversion of \eqref{eq5:what}, which reveals
\begin{equation}\label{eq5:detzn}
z_n = \frac{a_1 \Delta_{n+1} \underhat{\Delta}{}_n - a_2 
\Delta_n \underhat{\Delta}{}_{n+1}}{a_1 \kappa _2 \Delta_{n+1} \underhat{\Delta}{}_n -q a_2 \kappa _1 \Delta_n \underhat{\Delta}{}_{n+1}}.
\end{equation}
These may correspond to known determinantal solutions, such as the
Casorati determinants of Sakai \cite{Sakai:DetSolsqP6}, although we are yet to investigate this 
point.

\subsection{B\"acklund transformations}

The linear problem equivalent to the orthogonal polynomials three term recursion, \eqref{eq5:ynstarevoln}, may be expressed in terms of the natural variables appearing in \eqref{eq5:bn}. Substitution of \eqref{eq5:bn} into \eqref{eq5:Mnstar} gives
\[
M_n^* =  \left(
\begin{array}{cc}
{\displaystyle \frac{q^2 \kappa _1 (x-\alpha)+\kappa _2 (q \beta-x)}{q^3 \kappa _1-q \kappa _2}} & {\displaystyle \frac{\kappa _2 w_n}{q \kappa _1-\kappa _2} } \\
 {\displaystyle \frac{\kappa _2-q \kappa _1}{q \kappa _2 w_n}}  & 0
\end{array}
\right).
\]
In the context of orthogonal polynomial theory the system of equations describing the evolution of this system in the $n$ direction are known the Laguerre-Freud equations. Moreover, these very recurrence relations in the transformation $n \to n-1$ and $n \to n+1$ represent elements in the group of B\"acklund transformations. Since the group of B\"acklund transformations are of affine Weyl type, the Laguerre-Freud equations are equivalent to a translational component of the extended affine Weyl group of type $D_5^{(1)}$. 
We represent the $n \to n-1$ translation as 
\begin{equation}\label{5.41a}
\left\{ \begin{array}{c c c c} a_1 & a_2 & a_3 & a_4 \\
a_5 & a_6 & a_7 & a_8 \end{array} : y_n \,\, z_n\right\} \to \left\{ \begin{array}{c c c c} a_1 & a_2 & a_3 & a_4 \\
a_5 & a_6 & qa_7 &  \frac{a_8}{q} \end{array} : y_{n-1} \,\, z_{n-1}\right\}.
\end{equation}
The derivation of its explicit form relies on \eqref{eq5:compxn}. The lower right entry of \eqref{eq5:compxn} shifted $n \to n-1$ at $x = y_{n-1}$ yields the relation
\begin{equation}\label{eq4:yn-1}
y_{n-1} = -\frac{\delta}{\gamma}.
\end{equation}
By evaluating the upper right entry of \eqref{eq5:compxn} shifted by $n \to n-1$ at $x = y_{n-1}$ we obtain 
\begin{equation}
z_{n-1}=-\frac{(y_{n-1}-y_n) (y_{n-1}-\alpha)+z_{1}}{q \kappa _2 \left(a_4-y_{n-1}\right) \left(y_{n-1}-a_3\right)}.
\end{equation}
Finally, using \eqref{eq5:wndef} to find $w_{n-1}/w_n$ reveals
\[
w_{n-1} = \frac{w_n(\kappa_1- q\kappa_2)}{a_n^2(q\kappa_1 - \kappa_2)},
\]
which expresses $y_{n-1}$, $z_{n-1}$ and $w_{n-1}$ in terms of $y_{n}$, $z_n$ and $w_n$. 

A more canonical transformation from the orthogonal polynomial perspective is the transformation corresponding to the
shift $n \to n+1$, which is represented by
\[
\left\{ \begin{array}{c c c c} a_1 & a_2 & a_3 & a_4 \\
a_5 & a_6 & a_7 & a_8 \end{array} : y_n \,\, z_n\right\} \to \left\{ \begin{array}{c c c c} a_1 & a_2 & a_3 & a_4 \\
a_5 & a_6 & \frac{a_7}{q} &  q a_8 \end{array} : y_{n+1} \,\, z_{n+1}\right\} .
\]
Another viewpoint is that this shift is a $q$-difference analogue of a Schlesinger transformation of the linear system, which induces a B\"acklund transformation of the Painlev\'e equation\cite{JimboMiwa2}. The Schlesinger transformation is induced by multiplication on the left by a rational matrix, this rational matrix coincides with $M_n(x)$ for this particular solution of the linear system. 

\begin{theorem}
The shift $(y_n, z_n) \to (y_{n+1},z_{n+1})$ is given by

\begin{equation}
\label{znp1} z_{n+1} = \frac{\kappa _2 z_n \left[y_n(a_1 t -y_n)+\zeta_n\right] \left[y_n(a_2 t-y_n)+\zeta_n\right]}{q^2 \kappa _1 \left[\kappa _2 y_n z_n \left(a_3-y_n\right)+\zeta_n\right] \left[\kappa _2 y_n z_n \left(a_4-y_n\right)+\zeta_n\right]},
\end{equation}
\begin{align}
\label{ynp1} y_{n+1}&= \frac{\kappa _2 y_n \left(1-\kappa _2 z_n\right)}{q^2 \kappa _1 \left(1-q^2 \kappa
   _1 z_{n+1}\right)} \\ \times &\left[ \frac{ \zeta_n -(y_n-a_1t)(y_n-a_2t)+ {\displaystyle \frac{  z_n(q\theta_1t-\kappa_2a_1a_2t^2)}{1-\kappa_2z_n}}}{\kappa _2 y_n z_n \left(a_3-y_n\right)+\zeta _n }\right]  
 \left[ \frac{ \zeta_n -(y_n-a_1t)(y_n-a_2t)+ {\displaystyle \frac{z_n(q\theta_2t-\kappa_2a_1a_2t^2)}{1-\kappa_2z_n}}}{\kappa _2 y_n z_n \left(a_4-y_n\right)+\zeta _n} \right],\nonumber
\end{align}
where
\begin{multline*}
(\kappa_2 - q^2\kappa_1)\zeta_n = \kappa_2(y_n-a_1t)(y_n-a_2t)- q^2 \kappa_1\kappa_2(y_n-a_3)(y_n-a_4)z_n\\ + \frac{\kappa_2z_n}{1-\kappa_2z_n} \frac{(t\theta_1 - q\kappa_1a_3a_4)(t\theta_2 - q\kappa_1a_3a_4)}{\kappa_1a_3a_4} .
\end{multline*}
\end{theorem}
\begin{proof}
Using \eqref{eq5:compxn}, we note that an alternate way of writing $A_{n+1}^*$ is given by
\begin{equation}\label{2reps}
A_{n+1}^*(x,t) = M_n^*(qx,t)A_n^*(x,t)\left(M_n^*(x,t)\right)^{-1}.
\end{equation}
Using \eqref{eq5:AS} to represent the top row, and the right hand side of \eqref{2reps} to express the bottom row, we have
\[
A_{n+1}^* = \begin{pmatrix} q\kappa_1((x-y_{n+1})(x-\tilde{\alpha}) + \tilde{z}_1) & q^{-1}\kappa_2 w_{n+1}(x-y_{n+1})\\
{\displaystyle -\frac{\left(\kappa _2-q \kappa _1\right){}^2 (x-y_n)}{q \kappa _2 w_n}} & \left(x-y_n\right) \left(b_n \left( \kappa _1-q^{-1}\kappa _2\right)- \kappa _1 \alpha +q^{-1}x \kappa _2\right)+ z_1 \kappa _1
\end{pmatrix},
\]
where $\tilde{z}_1$ and $\tilde{\alpha}$ denotes $z_1$ and $\alpha$ at $n + 1$. The determinant of $A_{n+1}^*$ at $x = a_1t$ is zero. However, using this representation of $A_{n+1}^*$, the top row is divisible by $(y_{n+1} - a_1t)$ and the bottom row is divisible by $(y_n - a_1t)$. This also applies to the case for $x= a_2t$, hence by equating the determinant of $A_{n+1}^*$ with zero gives two expressions for $w_{n+1}$
\begin{align}
\label{T1}-\left(\frac{y_{n+1}-a_2 t}{z_{n+1}}-q^2 \kappa _1 \left(a_1 t-\tilde{\alpha}\right)\right)\left( \frac{y_n-a_2 t}{ z_n}-\left(q \kappa _1-\kappa _2\right) b_n+q \kappa _1 \alpha -a_1 t \kappa _2\right) = \frac{\left(\kappa _2-q \kappa _1\right){}^2 w_{n+1}}{w_n},\\
\label{T2}-\left(\frac{y_{n+1}-a_1 t}{z_{n+1}}-q^2 \kappa _1 \left(a_2 t-\tilde{\alpha}\right)\right)\left( \frac{y_n-a_1 t}{ z_n}-\left(q \kappa _1-\kappa _2\right) b_n+q \kappa _1 \alpha-a_2 t \kappa _2\right) = \frac{\left(\kappa _2-q \kappa _1\right){}^2 w_{n+1}}{w_n}.
\end{align}
In a similar manner, we consider the matrix representation of $A_{n+1}^*$ given by
\[
A_{n+1}^* = \begin{pmatrix} (x-y_n) \left(\kappa _1 (q x-b_n)+q^{-1}\kappa _2 (b_n-q \beta)\right)+\kappa _2 z_2 & q^{-1}\kappa_2 w_{n+1}(x-y_{n+1})\\
{\displaystyle -\frac{\left(\kappa _2-q \kappa _1\right){}^2 (x-y_n)}{q \kappa _2 w_n}} & {\displaystyle \frac{\kappa_2}{q}((x-y_{n+1})(x-\tilde{\beta}) + \tilde{z}_2)}
\end{pmatrix},
\]
which has been obtained by using the left hand side of \eqref{2reps} to represent the left column of $A_{n+1}^*$ and the right hand side of \eqref{2reps} to represent the right column of $A_{n+1}^*$. The left and right columns are divisible by $y_n- a_3$ and $y_{n+1}-a_3$ respectively at $x= a_3$. This applies also in the case of $x= a_4$. Hence equating the determinant of this representation of $A_{n+1}^*$ at $x=a_3$ and $x=a_4$ with zero gives two additional equations for $w_{n+1}$
\begin{align}
\label{T3}-\left( z_n(a_4 - y_n) + \frac{b_n}{q}\left(\frac{1}{q\kappa_1}- \frac{1}{\kappa_2}\right)- \frac{\beta}{q\kappa_1} + \frac{a_3}{\kappa_2} \right)\left(z_{n+1} \left(a_4-y_{n+1}\right)+\frac{ a_3-\tilde{\beta}}{q^2\kappa_1}\right) = \frac{1}{q^2}\left(\frac{1}{q\kappa_1} - \frac{1}{\kappa_2}\right)^2 \frac{w_{n+1}}{w_n},\\
\label{T4}-\left(z_n(a_3 - y_n) + \frac{b_n}{q}\left(\frac{1}{q\kappa_1}- \frac{1}{\kappa_2}\right)- \frac{\beta}{q\kappa_1} + \frac{a_4}{\kappa_2}  \right) \left(z_{n+1} \left(a_3-y_{n+1}\right)+\frac{ a_4-\tilde{\beta}}{q^2\kappa_1}\right) = \frac{1}{q^2}\left(\frac{1}{q\kappa_1} - \frac{1}{\kappa_2}\right)^2 \frac{w_{n+1}}{w_n}.
\end{align}
Equating the coefficient of $x$ in the upper right entry of \eqref{eq5:compxn} with zero reveals 
\[
\tilde{\alpha} + y_{n+1} = \frac{\left(q \kappa _1-\kappa _2\right) b_n+q \left(q \kappa _1 y_n+\kappa _2 \beta\right)}{q^2 \kappa _1},
\]
which reduces (\ref{T1}-\ref{T4}) to expressions for $w_{n+1}$ that are all of degree one in $y_{n+1}$ and $z_{n+1}$.
The compatibility between \eqref{T1} and \eqref{T2} is equivalent to
\begin{align}
\label{ynp1p} y_{n+1}z_n \left(\kappa _2-q^2 \kappa _1\right)\left(q^2 \kappa _1 z_{n+1}-1\right) =& \frac{a_3 a_4 \kappa _1 \kappa _2 \left(\kappa _2 z_n-q^2 \kappa _1 z_{n+1}\right) \left(a_1 a_2 t-q \theta _1 z_n\right) \left(a_1 a_2 t-q \theta _2 z_n\right)}{\theta _1 \theta _2 y_n
   \left(\kappa _2 z_n-1\right)} \\
  & -  q^2 \kappa _1 \left(z_n-z_{n+1}\right) \left(\left(a_3+a_4\right)
   \kappa _2 z_n-\left(a_1+a_2\right) t\right)\nonumber \\ & +q^2 y_n \kappa _1 \left(\kappa _2 z_n-1\right)
   \left(z_n-z_{n+1}\right).\nonumber
\end{align}
as is that of \eqref{T3} and \eqref{T4}. Substituting \eqref{ynp1p} into the equation resulting from the comparison of \eqref{T1} and \eqref{T3} the yields \eqref{znp1}. To obtain \eqref{ynp1}, we substitute the expressions for $z_{n+1}$, given by \eqref{znp1}, into the right hand side of \eqref{ynp1p}.
\end{proof}

As a preliminary check of the recurrence relations, we may consider the special case in which the support is chosen to be between $0$ and $qa_1t$. In this special case the moments are 
\[
\mu_k = \frac{(qa_1t)^{\sigma+k+1} (1-q)\left( \frac{a_1q^{\sigma+k+2}}{a_2},q;q\right)_\infty}{\left(q^{\sigma+k+1},\frac{qa_1}{a_2};q\right)_{\infty}} {}_2\phi_1 \left(\begin{array}{c |} \frac{a_4}{a_3} , q^{\sigma+k+1} \\ \frac{a_1q^{\sigma+k+2}}{a_2} \end{array} \hspace{.1cm} q; \frac{qa_1t}{a_4} \right).
\]
We explicitly compute the eigenvalues of $A_0^*$ to be
\[
\theta_1 = q^\sigma a_1a_2a_3a_4, \,\,  \theta_2 = a_1 a_2 a_3 a_4.
\]

Substituting these values of $\mu_k$ into \eqref{eq5:detyn}, \eqref{eq5:detzn}, \eqref{eq2:Delta} and \eqref{eq2:Sigma} for the $n=0$ case gives us the seed solution
\begin{align*}
y_0 &=  \frac{a_2 a_4 \left(a_1 t+a_3\right)-a_1 a_3 \left(a_2 t+a_4\right) q^{\sigma +1}}{a_2 a_4-a_1 a_3 q^{\sigma +1}}-\frac{a_1 a_2 t \left(q^{\sigma +1}-1\right)
   \left(a_1 a_3 q^{\sigma +2}-a_2 a_4\right) {}_2\phi_1 \left(\begin{array}{c |} \frac{a_4}{a_3} , q^{\sigma+2} \\ \frac{a_1q^{\sigma+3}}{a_2} \end{array} \hspace{.1cm} q; \frac{qa_1t}{a_4} \right) }{\left(a_1 q^{\sigma
   +2}-a_2\right) \left(a_1 a_3 q^{\sigma +1}-a_2 a_4\right){}_2\phi_1 \left(\begin{array}{c |} \frac{a_4}{a_3} , q^{\sigma+1} \\ \frac{a_1q^{\sigma+2}}{a_2} \end{array} \hspace{.1cm} q; \frac{qa_1t}{a_4} \right) },\\ 
   z_0 &= \frac{a_2 q^{-\sigma -1} {}_2\phi_1 \left(\begin{array}{c |} \frac{a_4}{a_3} , q^{\sigma+1} \\ \frac{a_1q^{\sigma+2}}{a_2} \end{array} \hspace{.1cm} q; \frac{a_1t}{a_4} \right) -a_1\,\,  {}_2\phi_1 \left(\begin{array}{c |} \frac{a_4}{a_3} , q^{\sigma+1} \\ \frac{a_1q^{\sigma+2}}{a_2} \end{array} \hspace{.1cm} q; \frac{qa_1t}{a_4} \right) }{a_1 a_2 a_3\,\,  {}_2\phi_1 \left(\begin{array}{c |} \frac{a_4}{a_3} , q^{\sigma+1} \\ \frac{a_1q^{\sigma+2}}{a_2} \end{array} \hspace{.1cm} q; \frac{a_1t}{a_4} \right) -a_1 a_2 a_4 \,\, {}_2\phi_1 \left(\begin{array}{c |} \frac{a_4}{a_3} , q^{\sigma+1} \\ \frac{a_1q^{\sigma+2}}{a_2} \end{array} \hspace{.1cm} q; \frac{qa_1t}{a_4} \right) }.
\end{align*}

As an illustration of the computation content of the recurrence relations and as a check on their veracity, we may compare numerical values of $y_{n+1}$ and $z_{n+1}$ using \eqref{eq5:detyn}, \eqref{eq5:detzn}, \eqref{eq2:Delta} and \eqref{eq2:Sigma} found by using \eqref{ynp1} and \eqref{znp1} from $y_n$ and $z_n$ for generic values of the parameters, $t$ and small values of $n$. Numerical evidence has been obtained to verify that $(y_1,z_1)$, found using \eqref{eq5:detyn}, \eqref{eq5:detzn}, \eqref{eq2:Delta} and \eqref{eq2:Sigma}, coincides with the values of $(y_1,z_1)$ found by using \eqref{ynp1} and \eqref{znp1} from the values of $(y_0,z_0)$ and \eqref{eq5:detyn}, \eqref{eq5:detzn}, \eqref{eq2:Delta} and \eqref{eq2:Sigma}. In a similar manner, we were also able to test the relationship between $(y_1,z_1)$ and $(y_2,z_2)$ using \eqref{ynp1} and \eqref{znp1} compared with values obtained by using \eqref{eq5:detyn}, \eqref{eq5:detzn}, \eqref{eq2:Delta} and \eqref{eq2:Sigma}.

We remark that the evolution $n \to n+1$ of the linear system corresponding to a deformed version of the Pastro weight supported on the unit circle, which is the circular analogue of the little $q$-Jacobi weight, has recently been obtained by Biane \cite{Biane:qP6}. The structure of the iterations in $n$ should have a similar structure to other translational components of the affine Weyl group, such as the translational component that coincides with the evolution of $q$-$\mathrm{P}_{\mathrm{VI}}$. This multiplicative structure is of \eqref{ynp1} and \eqref{znp1} is similar to the B\"acklund transformation of Biane \cite{Biane:qP6}. In the work of Biane \cite{Biane:qP6}, the B\"acklund transformation, representing the shift $n \to n+1$, simultaneously changes one of the eigenvalues of $A_{0,n}^*$ and $A_{2,n}^*$, whereas in our transformation, the eigenvalues of $A_{0,n}^*$ are independent of $n$. 
The little $q$-Jacobi case has also been studied in \cite{GK_2009}, although in a truncated way. It is clear from this work
that the authors have treated a specialized, in the sense that $ t $ is fixed by the parameters, and a degenerate case, whereby
the parameters are related by $ a_1a_4 = a_2a_3 $, and consequently have recovered elementary function
expressions for the three-term recurrence coefficients.   

\section*{Acknowledgments}

This research was supported in part by the Australian Research Council grant \#DP0881415.

\end{document}